\documentclass[letterpaper,11pt]{article}
\usepackage[top=1in, bottom=1in, left=1in, right=1in]{geometry}
\usepackage[T1]{fontenc}
\usepackage[dvipsnames]{xcolor}
\usepackage{hyperref}
\usepackage{microtype}
\usepackage{mathtools,xspace,amsthm,amsfonts,amssymb}
\usepackage[inline]{enumitem}
\usepackage[ruled,vlined,linesnumbered]{algorithm2e}
\SetKwComment{Comment}{$\triangleright$\ }{}
\date{}

\theoremstyle{plain}
\newtheorem{theorem}{Theorem}

\newtheorem{lemma}[theorem]{Lemma}

\newtheorem{observation}[theorem]{Observation}
\theoremstyle{definition}
\newtheorem{definition}[theorem]{Definition}

\newif\ifComments
\Commentstrue

\ifComments
    \newcommand{\mdb}[1]{\textcolor{blue}{Mark: #1}}
    \newcommand{\mm}[1]{\textcolor{magenta}{Mehran: #1}}
    \newcommand{\skb}[1]{\textcolor{red}{S\'andor: #1}}
\else
    \newcommand{\mdb}[1]{}
    \newcommand{\mm}[1]{}
    \newcommand{\skb}[1]{}
\fi

\newcommand{\myclaimWithProof}[2]{\begin{quotation} \noindent {\emph{Claim.} }{#1} \\[2mm] {\emph{Proof of claim.}} #2 \hfill {\footnotesize $\triangleleft$} \end{quotation}}

\newcommand{\mypara}[1]{\medskip\noindent\textbf{#1}}

\newcommand{\ngain}{{\alpha}}
\newcommand{\ngainb}{{\beta}}
\newcommand{\gain}{{a}}
\newcommand{\gainb}{{b}}
\newcommand{\fa}{{A}}
\newcommand{\fb}{{B}}
\newcommand{\fg}{\mathrm{gain}}

\newcommand{\fmu}{M}

\newcommand{\jmu}{j}
\newcommand{\mmu}{m}
\newcommand{\fgain}{{\mbox{gain}_\alpha}}
\newcommand{\fgainb}{{\mbox{gain}_\beta}}

\newcommand{\fy}{X_{\max}}
\newcommand{\vy}{y}
\newcommand{\gainplus}{\oplus}
\newcommand{\ftau}{{\Sigma_{\mathrm{gain}}}}

\newcommand{\NA}{\mbox{\sc Nil}}
\newcommand{\strict}{\mathrm{strict}}
\newcommand{\loose}{\mathrm{loose}}

\newcommand{\Reals}{{\mathbb{R}}}

\newcommand{\PlQ}{${\mathcal Q}$\xspace}
\newcommand{\PlP}{${\mathcal P}$\xspace}
\newcommand{\etal}{\emph{et al.}\xspace}
\newcommand{\eps}{\varepsilon}
\newcommand{\fs}{\widehat{s}}
\newcommand{\ft}{\widehat{t}}
\newcommand{\ComputeSolutions}{\mbox{{\sc ComputeSolutions}}\xspace}
\DeclareMathOperator{\vdeg}{deg}

\DeclareMathOperator{\distop}{dist}

\DeclareMathOperator{\fspan}{span}
\DeclareMathOperator*{\argmax}{argmax}
\DeclareMathOperator*{\argmin}{argmin}

\renewcommand{\leq}{\leqslant}
\renewcommand{\geq}{\geqslant}

\newcommand{\BeginMyItemize}{\begin{itemize}\setlength{\itemsep}{-\parskip}}
\newcommand{\EndMyItemize}{\end{itemize}}
\newcommand{\myitemize}[1]{\BeginMyItemize #1 \EndMyItemize}

\newcommand{\BeginMyEnumerate}{\begin{enumerate}\setlength{\itemsep}{-\parskip}}
\newcommand{\EndMyEnumerate}{\end{enumerate}}
\newcommand{\myenumerate}[1]{\BeginMyEnumerate #1 \EndMyEnumerate}

\newcommand{\win}{\mathit{win}}
\newcommand{\maj}{\mathit{Majority}}
\newcommand{\PlanarQuantCNF}{{\sc Planar }$\forall\exists$3-{\sc cnf}\xspace}

\title{On One-Round Discrete Voronoi Games\footnote{This work was partially supported by the Netherlands Organization for Scientific Research NWO under project no. 024.002.003.}
}
\author{\addtocounter{footnote}{2}
Mark de Berg\thanks{Department of Mathematics and Computer Science, TU Eindhoven, the Netherlands.
     Email: {\tt mdberg@win.tue.nl, s.kisfaludi.bak@tue.nl, mehran.mehr@gmail.com}} \and
\and
S\'andor Kisfaludi-Bak$^{\mathsection}$
\and
Mehran Mehr$^{\mathsection}$
}

\begin{document}
\maketitle

\begin{abstract}
Let $V$ be a multiset of $n$ points in $\Reals^d$, which we call voters,
and let $k\geq 1$ and $\ell\geq 1$ be two given constants.
We consider the following game, where two players \PlP and \PlQ compete
over the voters in $V$: First, player \PlP selects $k$~points in $\Reals^d$,
and then player \PlQ selects $\ell$~points in $\Reals^d$.
Player \PlP wins a voter
$v\in V$ iff $\distop(v,P) \leq \distop(v,Q)$, where
$\distop(v,P) :=  \min_{p\in P} \distop(v,p)$ and $\distop(v,Q)$ is defined similarly.
Player \PlP wins the game if he wins at least half the voters.
The algorithmic problem we study is the following: given $V$, $k$, and $\ell$,
how efficiently can we decide if player \PlP has a winning strategy, that is,
if \PlP can select his $k$ points such that he wins the game no matter where \PlQ places her points.

Banik~\etal devised a singly-exponential algorithm for the game in $\Reals^1$,
for the case~$k=\ell$. We improve their result by presenting
the first polynomial-time algorithm for the game in $\Reals^1$. Our algorithm
can handle arbitrary values of $k$ and $\ell$. We also show that if $d\ge 2$,
deciding if player \PlP has a winning strategy is $\Sigma_2^P$-hard
when $k$ and $\ell$ are part of the input. Finally, we prove
that for any dimension~$d$, the problem is contained in the complexity
class~$\exists\forall \Reals$, and we give an algorithm that works in polynomial time
for fixed $k$ and $\ell$.
\end{abstract}

\section{Introduction}
Voronoi games, as introduced by Ahn~\etal~\cite{ahn2004competitive},
can be viewed as competitive facility-location problems in which two
players \PlP and \PlQ want to place their facilities in order to maximize their market area.
The Voronoi game of Ahn~\etal~is played in a bounded region $R\subset \Reals^2$,
and the facilities of the players are modeled as points in this region.
Each player gets the same number, $k$, of facilities, which they have to
place alternatingly. The market area of \PlP (and similarly of \PlQ)
is now given by the area of the region of all points $q\in R$ whose closest
facility was placed by \PlP, that is, it is the total area of the Voronoi cells
of \PlP's facilities in the Voronoi diagram of the facilities of \PlP and~\PlQ.
Ahn~\etal~proved that for $k>1$ and when the region~$R$ is a circle or a segment, the second player
can win the game by a payoff of $1/2+\eps$, for some $\eps>0$,
where the first player can ensure $\eps$ is arbitrarily small.

The one-round Voronoi game introduced by Cheong~\etal~\cite{cheong2004one} is similar
to the Voronoi game of Ahn~\etal, except that the first player must first place
all his $k$ facilities, after which the second player places all her~$k$ facilities.
They considered the problem where~$R$ is a square, and they showed
that when $k$ is large enough the first player can always win a fraction $1/2+\alpha$
of the area of~$R$ for some $\alpha>0$. Fekete and Meijer~\cite{fekete2005one} considered the problem on
a rectangle~$R$ of aspect ratio~$\rho \leq 1$. They showed that the first player
wins more than half the area of~$R$, unless $k \geq 3$ and $\rho > \sqrt{2}/n$,
or $k=2$ and $\rho > \sqrt{3}/2$. They also showed that if $R$ is a polygon with holes,
then computing the locations of the facilities for the second player that maximize the area she wins,
against a given set of facilities of the first player is NP-hard.

\mypara{One-round discrete Voronoi games.}
In this paper we are interested in \emph{discrete (Euclidean) one-round Voronoi games},
where the players do not compete for area but for a discrete set of points.
That is, instead of the region $R$ one is given a set $V$ of $n$
points in a geometric space, and a point $v\in V$ is won by the player
owning the facility closest to $v$.
(Another discrete variant of Voronoi games is played on graphs~\cite{spoerhase2009r,teramoto2011voronoi}
but we restrict our attention to the geometric variant.)
More formally, the problem we study is defined as follows.

Let $V$ be a multiset of $n$ points in $\Reals^d$, which we call
\emph{voters} from now on,
and let $k\geq 1$ and $\ell\geq 1$ be two integers.
The one-round discrete Voronoi game defined by the triple $\langle V,k,\ell\rangle$
is a single-turn game played between two players~\PlP and~\PlQ.
First, player \PlP places a set $P$
of $k$ points in $\Reals^d$, then player \PlQ places a set $Q$
of $\ell$ points in $\Reals^d$. (These points may coincide with the voters in~$V$.)
We call the set $P$ the \emph{strategy of \PlP} and the set $Q$ the \emph{strategy of \PlQ}.
Player~\PlP wins a voter $v\in V$ if $\distop(v,P) \leq \distop(v,Q)$, where
$\distop(v,P)$ and $\distop(v,Q)$ denote the minimum distance between a voter~$v$
and the sets $P$  and $Q$, respectively. Note that this definition
favors player~\PlP, since in case of a tie a voter is won by~\PlP.
We now define
$
V[P\succeq Q] := \{ v\in V : \distop(v,P) \leq \distop(v,Q) \}
$
to be the multiset of voters won by player \PlP when he uses strategy~$P$
and player \PlQ uses strategy~$Q$.
Player \PlP wins the game $\langle V,k,\ell\rangle$ if he wins at least half the voters in $V$, that is,
when $\big| V[P\succeq Q] \big| \geq n/2$; otherwise \PlQ wins the game.
Here $\big| V[P\succeq Q] \big|$ denotes the size of the multiset $V[P\succeq Q]$
(counting multiplicities).
We now define $\Gamma_{k,\ell}(V)$ as the maximum number of voters
that can be won by player \PlP against an optimal opponent:
\[
\Gamma_{k,\ell}(V) := \max_{P \subset \Reals^d,\; |P|=k} \ \ \min_{Q \subset \Reals^d,\; |Q|=\ell} \ \ \big| V[P\succeq Q] \big|.
\]
For a given multiset $V$ of voters, we want to decide if\footnote{One can also
require that $\Gamma_{k,\ell} (V) > n/2$; with some small modifications, all the
results in this paper can be applied to the case with strict inequality as well.}
$\Gamma_{k,\ell} (V) \geq n/2$.
In other words, we are interested in determining for a given game~$\langle V,k,\ell\rangle$
if \PlP has a \emph{winning strategy}, which is a set of $k$ points such that \PlP wins
the game no matter where \PlQ places her points.

An important special case, which has already been studied in spatial voting theory
for a long time, is when $k=\ell=1$~\cite{mckelvey1976voting}.
Here the coordinates of a point in $V$ represent the preference of
the voter on certain topics, and the point played by~\PlQ represents a certain
proposal. If the point played by \PlP wins against all possible points played
by \PlQ, then the \PlP's proposal will win the vote against any other
proposal. Note that in the problem definition we gave above, voters at equal
distance from~$P$ and~$Q$ are won by~\PlP, and \PlP has to win at least half the voters.
This is the definition typically used in papers of Voronoi
games~\cite{banik2013optimal,banik2013two,banik2017discrete,banik2016discrete}.
In voting theory other variants are studied as well, for instance
where points at equal distance to $P$ and $Q$ are not won by either of them,
and \PlP wins the game if he wins more voters than \PlQ;
see the paper by McKelvey and Wendell \cite{mckelvey1976voting}
who use the term \emph{majority points} for the former variant and
the term \emph{plurality points} for the latter variant.

\mypara{Previous work.}
Besides algorithmic problems concerning the one-round discrete Voronoi game
one can also consider combinatorial problems.
In particular, one can ask for bounds on $\Gamma_{k,\ell}(V)$ as a function of
$n$, $k$, and~$\ell$.
It is known that for any set $V$ in $\Reals^2$ and $k=\ell=1$ we have
$\left\lfloor n/3 \right\rfloor \leq \Gamma_{1,1}(V) \leq \left\lceil n/2 \right\rceil$.
This result is based on known bounds for maximum Tukey depth, where the lower bound
can be proven using Helly's theorem. It is also known~\cite{banik2016discrete} that
there is a constant $c$ such that $k=c\ell$ points suffice for \PlP to win the game,
that is, $\Gamma_{c\ell,\ell}(V) \geq n/2$ for any $V$.

In this paper we focus on the algorithmic problem of computing $\Gamma_{k,\ell}(V)$
for given $V$, $k$, and~$\ell$. The problem of deciding if $\Gamma_{k,\ell}(V) \geq n/2$
was studied for the case $k=\ell=1$ by Wu \etal \cite{wu2013computing} and Lin \etal \cite{lin2015forming}, and later by
De~Berg~\etal\cite{berg2018faster} who solve this problem in $O(n \log n)$ time
in any fixed dimension~$d$. Their algorithms works when $V$ is a set (not a multiset)
and for plurality points instead of majority points.
Other algorithmic results are for the setting where
the players already placed all but one of their points,
and one wants to compute the best locations for the last point of~\PlP and of~\PlQ.
Banik~\etal~\cite{banik2017discrete} gave algorithms that finds the best
location for \PlP in $O(n^8)$ time and for \PlQ in $O(n^2)$ time.
For the two-round variant of the problem, with $k=\ell=2$, polynomial algorithms for finding
the optimal strategies of both players are also known~\cite{banik2013two}.

Our work is inspired by the paper of Banik~\etal~\cite{banik2013optimal}
on computing $\Gamma_{k,\ell}(V)$ in~$\Reals^1$. They considered the case of arbitrarily
large $k$ and $\ell$, but where $k=\ell$ (and $V$ is a set instead of a multiset).
For this case they showed that depending on the set $V$ either \PlP or \PlQ
can win the game, and they presented an algorithm to compute $\Gamma_{k,\ell}(V)$
in time $O(n^{k-\lambda_k})$, where $0< \lambda_k <1$ is a constant dependent only on~$k$.
This raises the question: is the problem NP-hard when $k$ is part of the input?

\mypara{Our results.}
We answer the question above negatively, by presenting an algorithm that
computes $\Gamma_{k,\ell}(V)$ in~$\Reals^1$ in polynomial time. Our algorithm
works when $V$ is a multiset, and it does not require $k$ and $\ell$ to be
equal. Our algorithm computes $\Gamma_{k,\ell}(V)$ and finds a strategy for
\PlP that wins this many voters in time~$O(kn^4)$. The algorithm can be
extended to the case when the voters are weighted, requiring only a slight
increase in running time.

The algorithm by Banik~\etal~\cite{banik2013optimal} discretizes the problem,
by defining a finite set of potential locations for \PlP to place his points.
However, to ensure an optimal strategy for \PlP, the set of potential locations
has exponential size. To overcome this problem we need several new ideas.
First of all, we essentially partition the possible strategies into various
classes---the concept of thresholds introduced later plays this role---such
that for each class we can anticipate the behavior of the optimal strategy
for~\PlQ. To compute the best strategy within a certain class we use dynamic
programming, in a non-standard (and, unfortunately, rather complicated) way.
The subproblems in our dynamic-program are for smaller point sets and smaller
values of $k$ and $\ell$ (actually we will need several other parameters) where
the goal of \PlP will be to push his rightmost point as far to the right as possible to win a certain
number of points. One complication in the dynamic program is that it is
unclear which small subproblems~$I'$ can be used to solve a given subproblem~$I$.
The opposite direction---determining for $I'$ which larger subproblems~$I$
may use $I'$ in their solution---is easier however, so we use a sweep approach:
when the solution to some~$I'$ has been determined, we update the solution
to larger subproblems~$I$ that can use~$I'$.

After establishing that we can compute $\Gamma_{k,\ell}(V)$ in polynomial
time in~$\Reals^1$, we turn to the higher-dimensional problem. We show
that deciding if \PlP has a winning strategy is $\Sigma_2^P$-hard in $\Reals^2$.
We also show that for fixed $k$ and $\ell$ this problem can be solved
in polynomial time. Our solution combines algebraic methods~\cite{basu1996combinatorial}
with a result of Paterson and Zwick~\cite{paterson1993shallow} that one can construct
a polynomial-size boolean circuits that implements the majority function.
The latter result in essential to avoid the appearance of~$n$ in the exponent.
As a byproduct of the algebraic method, we show that the problem is contained
in the complexity class~$\exists\forall \Reals$; see \cite{dobbins2018area} for more information on this complexity class.

\section{A polynomial-time algorithm for \texorpdfstring{$d=1$}{d=1}}\label{sec:1d}
In this section we present a polynomial-time algorithm for the 1-dimensional discrete Voronoi game.
Our algorithm will employ dynamic programming, and it will be convenient to use $n$, $k$,
and $\ell$ as variables in the dynamic program. From now on, we therefore use
$n^*$ for the size of the original multiset $V$, and $k^*$ and $\ell^*$
for the initial number of points that can be played by \PlP and \PlQ, respectively.


\subsection{Notation and basic properties}
We denote the given multiset of voters by $V := \{v_1,\ldots,v_{n^*}\}$, where we assume
the voters are numbered from left to right. We also always number the points in the
strategies $P:= \{p_1,\ldots,p_{k^*}\}$ and $Q:= \{q_1,\ldots,q_{\ell^*}\}$ from left to right.
For brevity we make no distinction between a point and its value
(that is, its $x$-coordinate), so that we can for example write~$p_1 < q_1$ to indicate
 that the leftmost point
of~$P$ is located to the left of the leftmost point of~$Q$.

For a given game~$\langle V,k,\ell\rangle$, we
say that a strategy $P$ of player~\PlP \emph{realizes} a gain $\gamma$ if
$\big| V[P\succeq Q] \big| \geq \gamma$ for any strategy $Q$ of player~\PlQ.
Furthermore, we say that a strategy $P$ is \emph{optimal} if it realizes
$\Gamma_{k,\ell}(V)$, the maximum possible gain for \PlP, and
we say a strategy $Q$ is \emph{optimal} against
a given strategy~$P$ if $\big| V[P\succeq Q] \big| \leq \big| V[P\succeq Q'] \big|$
for any strategy~$Q'$.

\mypara{Trivial, reasonable, and canonical strategies for \PlP.}
For $0\leq n \leq n^*$, define $V_n\coloneqq \{v_1,\ldots,v_n\}$ to be the leftmost~$n$
points in $V$. Suppose we want to compute $\Gamma_{k,\ell}(V_n)$ for some
$1\leq k \leq n$ and  $0\leq \ell \leq n$.
The \emph{trivial strategy} of player~\PlP is to
place his points at the $k$ points of $V_n$ with the highest
multiplicities---here we consider the multiset $V_n$ as a set of distinct points, each
with a multiplicity corresponding to the number of times it occurs in~$V_n$---with
ties broken arbitrarily. Let $\| V_n\|$ denote the number of distinct points in $V_n$.
Then the trivial strategy is optimal
when $k\geq \| V_n \|$ and also when $\ell\geq 2k$: in the former case
\PlP wins all voters with the trivial strategy, and in the latter case \PlQ can
always win all voters not coinciding with a point in~$P$ (namely by surrounding each point
$p_i\in P$ by two points sufficiently close to $p_i$) so the trivial strategy is optimal for~\PlP.
Hence, from now on we consider subproblems with $k < \| V_n \|$ and $\ell < 2k$.


We can without loss of generality restrict our attention to strategies for \PlP
that place at most one point in each half-open interval of the form $(v_i,v_{i+1}]$
with $v_i\neq v_{i+1}$, where $0\leq i\leq n$, $v_0 \coloneqq -\infty$, and $v_{n^*+1} \coloneqq \infty$.
Indeed, placing more than two points inside an interval $(v_i,v_{i+1}]$
is clearly not useful, and if two points are placed in some interval~$(v_i,v_{i+1}]$
then we can always move the leftmost point onto~$v_i$. (If $v_i$ is already occupied by
a point in $P$, then we can just put the point on any unoccupied voter;
under our assumption that $k < \| V_n \|$ an unoccupied voter always exists.)
We will call a strategy for \PlP satisfying the property above \emph{reasonable}.

\begin{observation}[Banik~\etal~\cite{banik2013optimal}] \label{obs:banik}
Assuming $k < \|V_n\|$ there exist an optimal strategy for~\PlP
that is reasonable and has $p_1\in V$ (that is, $p_1$ coincides with a voter).
\end{observation}
We can define an ordering on strategies of the same size by sorting them in
lexicographical order. More precisely, we say that a strategy~$P =\{p_1,\ldots,p_k\}$
is \emph{greater than} a strategy~$P' =\{p'_1,\ldots,p'_k\}$, denoted by $P\succ P'$, if
$\langle p_1,\ldots,p_k\rangle >_{\mathrm{lex}} \langle p'_1,\ldots,p'_k\rangle$,
where $>_{\mathrm{lex}}$ denotes the lexicographical order.
Using this ordering, the largest reasonable strategy $P$
that is optimal---namely, that realizes $\Gamma_{k,\ell}(V_n)$---is called
the \emph{canonical strategy} of \PlP.

\mypara{$\ngain$-gains, $\ngainb$-gains, and gain sequences.}
Consider a strategy $P := \{p_1,\ldots,p_k\}$. It will be convenient to add two extra points
to~$P$, namely $p_0\coloneqq -\infty$ and $p_{k+1}\coloneqq \infty$; this clearly does not
influence the outcome of the game.
The strategy $P$ thus induces $k+1$ open intervals of the form $(p_i,p_{i+1})$ where player \PlQ may place her points.
It is easy to see that there exists an optimal strategy for \PlQ with the following property:
$Q$ contains at most two points in each interval $(p_i,p_{i+1})$
with $1\leq i\leq k-1$, and at most one point in $(p_0,p_{1})$ and
at most one point in~$(p_k,p_{k+1})$. From now on we restrict our attention to
strategies for \PlQ with this property.

Now suppose that $x$ and $y$ are consecutive points (with $x<y$) in some strategy~$P$,
where $x$ could be $-\infty$ and $y$ could be $\infty$. As just argued,
\PlQ either places zero, one, or two points inside $(x,y)$. When \PlQ places
zero points, then she obviously does not win any of the voters in $V_n \cap (x,y)$.
The maximum number of voters \PlQ can win from $V_n \cap (x, y)$ by placing
a single point is the maximum number of voters in $(x,y)$ that can be covered
by an open interval of length $(y-x)/2$; see~Banik~\etal~\cite{banik2013optimal}.
We call this value the \emph{$\ngain$-gain} of \PlQ  in $(x,y)$ and denote
it by $\fgain(V_n,x,y)$. By placing two points inside $(x,y)$, one
immediately to the right of $x$ and one immediately to the left of $y$,
player~\PlQ will win all voters $V_n \cap (x,y)$. Thus the extra number of voters
won by the second point in $(x,y)$ as compared to just placing a single point
is equal to $|V_n\cap (x,y)|-\fgain(V_n,x,y)$. We call this quantity
the \emph{$\ngainb$-gain} of $Q$ in $(x,y)$ and denote it by $\fgainb(V_n,x,y)$.
Note that for intervals $(x,\infty)$ we have $\fgain(x,\infty) = |V_n\cap (x,\infty)|$
and $\fgainb(x,\infty) = 0$; a similar statement holds for $(-\infty,y)$.

The following observation follows from the fact that
$\fgain(V_n,x,y)$ equals the maximum number of voters in $(x,y)$ that can be covered
by an open interval of length $(y-x)/2$.
\begin{observation}[Banik~\etal~\cite{banik2013optimal}]
\label{obs:lessgain}
For any $x,y$ we have $\fgain(V_n,x,y) \geq \fgainb(V_n,x,y)$.
\end{observation}
If we let $a := \fgain(V_n,x,y)$ and $b:=\fgainb(V_n,x,y)$, then player~\PlQ
wins either 0, $a$, or $a+b$ points depending on whether she plays 0, 1, or 2 points
inside the interval. It will therefore be convenient to introduce the notation
$\gainplus_j$ for $j\in\{0,1,2\}$, which is defined as
\[
\gain \gainplus_0 \gainb \coloneqq 0, \hspace*{5mm}
\gain \gainplus_1 \gainb \coloneqq \gain, \hspace*{5mm}
\gain \gainplus_2 \gainb \coloneqq \gain+\gainb. \\
\]
We assume the precedence of these operators are higher than addition.

Let $P:=\{p_0,p_1,\ldots,p_k,p_{k+1}\}$ be a given strategy for~\PlP,
where by convention $p_0=-\infty$ and $p_{k+1}=\infty$. Consider
$\{ \fgain(V_n,p_i,p_{i+1}) : 0\leq i\leq k \} \cup \{ \fgainb(V_n,p_i,p_{i+1}) : 0\leq i\leq k\}$,
the multiset of all $\ngain$-gains and $\ngainb$-gains defined by the intervals~$(p_i,p_i+1)$.
Sort this sequence in non-increasing order, using the following tie-breaking rules
if two gains are equal:
\myitemize{
\item if one of the gains is for an interval $(p_i,p_{i+1})$---that is, the gain is either 
$\fgain(V_n,p_i,p_{i+1})$ or $\fgainb(V_n,p_i,p_{i+1})$---and the other gain is
    for an interval $(p_j,p_{j+1})$ with $j>i$, then
    the gain for $(p_i,p_{i+1})$ precedes  the gain for $(p_j,p_{j+1})$.
\item if both gains are for the same interval $(p_i,p_{i+1})$ then the
      $\ngain$-gain precedes the $\ngainb$-gain.
}
We call the resulting sorted sequence the \emph{gain sequence}
induced by $P$ on~$V_n$. We denote this sequence by $\ftau(V_n,P)$ or, when
$P$ and $V_n$ are clear from the context, sometimes simply by~$\ftau$.

\mypara{The canonical strategy of~\PlQ and sequence representations.}
Given the multiset $V_n$, a strategy $P$ and value~$\ell$,
player \PlQ can compute an optimal strategy as follows.
First she computes the gain sequence~$\ftau(V_n,P)$ and
chooses the first $\ell$ gains in $\ftau(V_n,P)$. Then for each $0\leq i\leq k$
she proceeds as follows.
When $\fgain(V_n,p_i,p_{i+1})$ and $\fgainb(V_n,p_i,p_{i+1})$ are both chosen,
she places two points in $(p_i,p_{i+1})$ that win all voters in $(p_i,p_{i+1})$;
when only $\fgain(V_n,p_i,p_{i+1})$ is chosen,
she places one point in $(p_i,p_{i+1})$ that win $\fgain(V_n,p_i,p_{i+1})$ voters.
(By Observation~\ref{obs:lessgain} and the tie-breaking rules,
when $\fgainb(V_n,p_i,p_{i+1})$ is chosen it is always the case that $\fgain(V_n,p_i,p_{i+1})$
is also chosen.) The resulting optimal strategy $Q$ is called the \emph{canonical strategy}
of \PlQ with $\ell$ points against $P$ on $V_n$.

From now one we restrict the strategies of player~\PlQ to canonical
strategies. In an optimal strategy, player~\PlQ
places at most two points in any interval induced by a strategy $P=\{p_0,\ldots,p_{k+1}\}$,
and when we know that \PlQ places a single point (and similarly when she places two points)
then we also know where to place the point(s).
Hence, we can represent an optimal strategy $Q$, for given $V_n$ and $P$, by a
sequence~$\fmu(V,P,Q) \coloneqq \langle \mmu_0,\ldots,\mmu_k \rangle$
where $\mmu_i\in \{0,1,2\}$ indicates how many points \PlQ plays in
the interval~$(p_{i},p_{i+1})$. We call $\fmu(V,P,Q)$ the \emph{sequence representation}
of the strategy $Q$ against $P$ on $V_n$.
We denote the sequence representation of the canonical strategy of \PlQ with $\ell$ points
against $P$ on $V_n$ by $\fmu(V,P,\ell)$.
We have the following observation.
\begin{observation}
The canonical strategy of \PlQ with $\ell$ points against $P$ is the optimal strategy $Q$ with $\ell$ points against $P$ whose sequence representation is maximal in the lexicographical order.
\end{observation}

\subsection{The subproblems for a dynamic-programming solution} \label{sec:subproblem}
For clarity, in the rest of Section~\ref{sec:1d} we assume the multiset of voters $V$ does not have repetitive entries,
i.e we have a set of voters, and not a multiset. While all the results are easily extendible to multisets, dealing with them
adds unnecessary complexity to the text.

Our goal is to develop a dynamic-programming algorithm to compute~$\Gamma_{k^*,\ell^*}(V)$.
Before we can define the subproblems on which the dynamic program is based, we need to introduce
the concept of \emph{thresholds}, which is a crucial ingredient in the subproblems.

\mypara{Strict and loose thresholds.}
Consider an arbitrary gain sequence $\ftau(V_n,P) = \langle \tau_1,\ldots,\tau_{2k+2} \rangle$.
Recall that each $\tau_i$ is the $\ngain$-gain or $\ngainb$-gain of some interval~$(p_i,p_{i+1})$,
and that these gains are sorted in non-increasing order.
We call any integer value $\tau \in [\tau_{\ell+1},\tau_{\ell}]$ an
\emph{$\ell$-threshold} for \PlQ induced by $P$ on $V_n$,
or simply a \emph{threshold} if $\ell$ is clear from the context. We implicitly assume
$\tau_0\coloneqq n$ so that talking about $0$-threshold is also meaningful.
Note that when $\tau_{\ell}\geq\tau>\tau_{\ell+1}$ then the canonical strategy for~\PlQ
chooses all gains larger than $\tau$ and no gains smaller or equal to~$\tau$.
Hence, we call $\tau$ a \emph{strict} threshold if $\tau_{\ell}\geq\tau>\tau_{\ell+1}$.
On the other hand, when $\tau=\tau_{\ell+1}$ then gains of value
$\tau$ may or may not be chosen by the canonical strategy of \PlQ.
(Note that in this case for gains of value~$\tau$ to be picked, we would actually
need $\tau_{\ell}=\tau=\tau_{\ell+1}$.)
In this case we call $\tau$ a \emph{loose} threshold.
%
%

The idea will be to guess the threshold $\tau$ in an optimal solution
and then use the fact that fixing the threshold~$\tau$ helps us to
limit the strategies for \PlP and anticipate the behavior of \PlQ.
Let $P_{\mathrm{opt}}$ be the canonical strategy realizing $\Gamma_{k^*,\ell^*}(V)$.
We call any $\ell^*$-threshold of $P_{\mathrm{opt}}$
an \emph{optimal threshold}.
We devise an algorithm that gets a value $\tau$ as input and computes
$\Gamma_{k^*,\ell^*}(V)$ correctly if $\tau$ is an optimal threshold, and computes a value not greater
than $\Gamma_{k^*,\ell^*}(V)$, otherwise.

Clearly we only need to consider values of $\tau$ that are at most~$n^*$.
In fact, since each $\ngain$-gain or $\ngainb$-gain in a given
gain sequence corresponds to a unique subset of voters, the $\ell^*$-th largest
gain can be at most $n^*/\ell^*$, so we only need to consider $\tau$-values up to~$\lfloor n^*/\ell^*\rfloor$.
Observe that when there exists an optimal strategy that induces an $\ell^*$-threshold
equal to zero, then \PlQ can win all voters not explicitly covered by $P$.
In this case the trivial strategy is optimal for \PlP.
Our global algorithm is now as follows.
\myenumerate{
\item \label{step1} For all thresholds $\tau\in \{1,\ldots,\lfloor n^*/\ell^*\rfloor\}$,
      compute an upper bound on the number of voters \PlP can win with a strategy that has an $\ell^*$-threshold $\tau$.
      For the run where $\tau$ is an optimal threshold, the algorithm will return $\Gamma_{k^*,\ell^*}(V)$.
\item Compute the number of voters \PlP wins in the game $\langle V,k^*,\ell^*\rangle$
      by the trivial strategy.
\item Report the best of all solutions found.
}

\mypara{The subproblems for a fixed threshold~$\tau$.}
From now on we consider a fixed threshold value~$\tau\in \{1,\ldots,\lfloor n^*/\ell^*\rfloor\}$. The subproblems in our
dynamic-programming algorithm for the game~$\langle V,k^*,\ell^*\rangle$ have
several parameters.
\myitemize{
\item A parameter $n\in \{0,\ldots, n^*\}$, specifying that the subproblem is on the voter set $V_n$.
\item Parameters $k,\ell\in\{0,\ldots, n\}$, specifying that \PlP can use $k+1$ points and
      \PlQ can use $\ell$ points.
\item A parameter $\gamma\in\{0,\ldots, n\}$, specifying the number of voters \PlP must win.
\item A parameter $\delta\in\{\mathrm{strict, loose}\}$, specifying the strictness of the fixed $\ell$-threshold~$\tau$.
}
Intuitively, the subproblem specified by a tuple $\langle n,k,\ell,\gamma,\delta \rangle$ asks for a strategy~$P$
where \PlP wins at least $\gamma$ voters from $V_n$ and such that $P$ that induces an $\ell$-threshold of strictness~$\delta$,
against an opponent \PlQ using~$\ell$ points. Player~\PlP may use $k+1$ points and his objective will be to push
his last point, $p_{k+1}$ as far to the right as possible.  The value of the solution to such a subproblem,
which we denote by $\fy(n,k,\ell,\gamma,\delta)$, will indicate how far
to the right we can push $p_{k+1}$. Ultimately we will be interested in solutions where
\PlP can push $p_{k^*+1}$ all the way to $+\infty$, which means he can actually win~$\gamma$ voters
by placing only~$k^*$ points.
%

To formally define $\fy(n,k,\ell,\gamma,\delta)$, we need one final piece of notation.
Let $x\in \Reals \cup \{-\infty\}$, let $n\in\{1,\ldots, n^*\}$, and
let $\gain,\gainb$ be integers. For convenience, define $v_{n^*+1} \coloneqq \infty$.
Now we define the \emph{$(\gain,\gainb)$-span of $x$ to $v_{n+1}$},
denoted by $\fspan(x,n,\gain,\gainb)$, as
\begin{align*}
\mkern-36mu 
 \fspan&(x,n,\gain,\gainb) \coloneqq  \\
&
\mkern-30mu 
\begin{cases}
\parbox[c]{75mm}{the maximum real value $y\in(v_n,v_{n+1}]$ such that $\fgain(V,x,y) = \gain$ and $\fgainb(V,x,y) = \gainb$}
&  \mbox{if $x\neq -\infty$ and $y$ exists} \\
-\infty & \mbox{otherwise}
\end{cases}
\end{align*}
%
%
\begin{definition} \label{def:subproblem}
For parameters $n\in \{0,\ldots, n^*\}$, $k,\ell,\gamma\in\{0,\ldots, n\}$,
and $\delta\in\{\mathrm{strict, loose}\}$, we define the value
$\fy(n,k,\ell,\gamma,\delta)$ and what it means when a strategy $P$ \emph{realizes} this value, as follows.
\myitemize{
\item For $k=0$, we call it an \emph{elementary subproblem}, and define $\fy(n,k,\ell,\gamma,\delta) = v_{n+1}$ if
    \myenumerate{
    \item $\{v_{n+1}\}$ wins at least $\gamma$ voters from $V_{n}$, and
    \item $\{v_{n+1}\}$ induces an $\ell$-threshold $\tau$ with strictness $\delta$ on $V_n$,
    } 
and we define $\fy(n,k,\ell,\gamma,\delta) = -\infty$ otherwise.
In the former case we say that $P := \{v_{n+1}\}$ \emph{realizes} $\fy(n,k,\ell,\gamma,\delta)$. 

\bigskip

\item For $k>0$, we call it a \emph{non-elementary subproblem}, and $\fy(n,k,\ell,\gamma,\delta)$ is defined to be equal to
the maximum real value $y\in(v_n,v_{n+1}]$ such that there exists a strategy $P\coloneqq P' \cup \{\vy\}$ with $P'=\{p_1,\ldots,p_{k}\}$, integer values $n',\gain,\gainb$ with  $0\leq n' < n$ and $0\leq \gain,\gainb \leq n$, an integer $\jmu\in\{0,1,2\}$, and a $\delta'\in\{\strict,\loose\}$  satisfying the following conditions:
\myenumerate{
\item $P$ wins at least $\gamma$ voters from $V_{n}$,
\item $P$ induces an $\ell$-threshold $\tau$ with strictness $\delta$ on $V_n$,
\item $\fspan(p_k,n,\gain,\gainb) = \vy$,
\item $P'$ realizes $\fy(n',k-1,\ell-\jmu,\gamma-n+n'+\gain\gainplus_j\gainb,\delta')$,
\item Let $\fmu(V_{n'},P',\ell-j) \coloneqq \langle \mmu'_0,\ldots,\mmu'_k \rangle$ and $\fmu(V_n,P,\ell) \coloneqq \langle \mmu_0,\ldots,\mmu_{k+1} \rangle$. Then
   $\mmu'_i = \mmu_i$ for all $0\leq i<k$.
}  
When a set $P$ satisfying the conditions exists, we say that $P$ \emph{realizes} $\fy(n,k,\ell,\gamma,\delta)$. We define $\fy(n,k,\ell,\gamma,\delta) = -\infty$ if no such $P$ exists.
} 
\end{definition}
By induction we can show that if the parameters $n,k,\ell$ are not in a certain range, namely
if one of the conditions $\ell < 2(k+1)$ or $k \leq \|V_n\|$ is violated, then $\fy(n,k,\ell,\gamma,\delta)=-\infty$.
%
%
%
The next lemma shows we can compute $\Gamma_{k^*,\ell^*}(V)$
from the solutions to our subproblems.
\begin{lemma}
\label{lem:computeGamma}
Let $V=\{v_1,\ldots,v_{n^*}\}$ be a set of $n^*$ voters in $\Reals^1$. 
Let $0\leq k^*\leq n^*$ and $1\leq\ell^*\leq n^*$ be two integers such that
$\ell^* < 2k^*$ and $k^* < \|V\|$, and let $\tau$ be a fixed threshold. Then
\begin{equation} \label{eq:lem:1}
\Gamma_{k^*,\ell^*}(V) \ \geq \ \parbox[t]{100mm}{the maximum value of $\gamma$
with $0\leq \gamma \leq n^*$ for which there exist a $\delta\in\{\loose,\strict\}$ such that
$\fy(n^*,k^*,\ell^*,\gamma,\delta) = \infty $.}
\end{equation}
Moreover, for an optimal threshold $\tau_{\mathrm{opt}}>0$, the inequality changes to equality.
\end{lemma}
\begin{proof}
It is clear that $\Gamma_{k^*,\ell^*}(V)$ is at least equal to the right-hand side in~\eqref{eq:lem:1}. Indeed, an equation $\fy(n^*,k^*,\ell^*,\gamma,\delta)=\infty$
implies by definition that there is a strategy $P'$ of $k^*$ points such that
$P'\cup \{\infty\}$ wins at least $\gamma$~voters against an opponent~\PlQ
with $\ell^*$ points, and then $P'$ must win $\gamma$ voters as well.
Next we prove the opposite direction for an optimal threshold~$\tau_{\mathrm{opt}}$.

Let $P_{\mathrm{opt}}$ be the canonical strategy realizing $\Gamma_{k^*,\ell^*}(V)$. By definition, $\tau_{\mathrm{opt}}$ is an $\ell^*$-threshold induced by $P_{\mathrm{opt}}$ on $V$.
Let $Q_{\mathrm{opt}}$
be the canonical strategy of \PlQ with $\ell^*$ points against $P_{\mathrm{opt}}$ on $V$
and let $\fmu(V,P_{\mathrm{opt}},\ell^*) = \langle m_0,\ldots,m_{k^*}\rangle$ be its sequence representation.
Define $P_k := \{p_1,\ldots,p_k,p_{k+1}\}$ and $Q_k := Q_{\mathrm{opt}}\cap [-\infty,p_{k+1}]$,
and let $n_k$ be the largest index such that~$v_{n_k} < p_{k+1}$.
Let $\ell_k = |Q_k|$, and $\gamma_k= |V_{n_k}[P_k\succeq Q_k]|$, that is,
$\gamma_k$ is the number of voters from $V_{n_k}$ won by $P_k$ against~$Q_k$.
Note that $\tau_{\mathrm{opt}}$ is an $\ell_k$-threshold of $P_k$ on $V_{n_k}$.
Let $\delta_k\in\{\strict,\loose\}$ indicate whether our fixed threshold~$\tau_{\mathrm{opt}}$ is
a strict or loose threshold induced by $P_k$ and $\ell_k$ on $V_{n_k}$,
where $p_{k^*+1} \coloneqq \infty$.

We will need the following claim.
\myclaimWithProof{
$Q_k$ is the canonical strategy of \PlQ with $\ell_k$ points against $P_k$ on $V_{n_k}$, and
for all $0 \leq k \leq k^*$ we have
$\fmu(V_{n_k},P_k,Q_k) = \langle \mmu_0,\ldots,\mmu_k, 0\rangle$.
}%
{The second part of the claim is obvious by the definition of sequence representation and by the definition of $Q_k$.

The first part of the claim can be shown by contradiction. Suppose $Q'_k$ is the canonical strategy of \PlQ with $\ell_k$ points against $P_k$ on $V_{n_k}$, where $Q'_k \neq Q_k$. Let $\fmu(V_{n_k},P_k,Q'_k)=\langle \mmu'_0,\ldots,\mmu'_k, 0 \rangle$. Consider the strategy $Q$ with sequence representation $\langle \mmu'_0,\ldots,\mmu'_k, \mmu_{k+1}, \ldots, \mmu_{k^*} \rangle$ of size $\ell^*$ against $P_{\mathrm{opt}}$.
We have two cases.
\begin{itemize}
\item $Q_k$ is not an optimal strategy against $P_k$. This means $Q$ wins more voters than $Q_{\mathrm{opt}}$ against $P_{\mathrm{opt}}$, because the winning of any given voter by \PlQ is only dependent on the number of points \PlQ places in the interval $(p_i,p_{i+1})$ containing that voter. But this contradicts the fact that $Q_{\mathrm{opt}}$ is an optimal strategy against $P_{\mathrm{opt}}$.
\item $Q_k$ is an optimal strategy against $P_k$, but $Q_k$ is smaller in the lexicographical
      order than $Q'_k$. This means that $Q$ is a strategy of equal gain as $Q_{\mathrm{opt}}$ whose sequence representation is lexicographically larger. But this contradicts the fact that $Q_{\mathrm{opt}}$ is a canonical strategy against $P_{\mathrm{opt}}$.
\end{itemize}
}
Let $I_k$ denote the subproblem given by $\langle n_k,k,\ell_k,\gamma_k,\delta_k\rangle$
for $\tau=\tau_{\mathrm{opt}}$.
Below we will show by induction on $k$ that for any $0\leq k\leq k^*$ we have
\begin{equation}
\label{eq:secondeq}
\fy(I_k) = p_{k+1} \mbox{ and $\fy(I_k)$ is realized by $P_k$},
\end{equation}
where we use $\fy(I_k)$ as a shorthand for $\fy(n_k,k,\ell_k,\gamma_k,\delta_k)$.
Note that~(\ref{eq:secondeq}) finishes the proof of the lemma. Indeed, $p_{k^*+1}=\infty$ by definition,
and $n_{k^*}=n^*$ which implies $V_{n_{k^*}}=V$. Hence,
$\gamma_{k^*} = \Gamma_{k^*,\ell^*}(V)$ and $\ell_{k^*}=\ell^*$,
and so there is a $\delta\in\{\loose,\strict\}$ such that $\fy(n^*,k^*,\ell^*,\Gamma_{k^*,\ell^*}(V),\delta_{k^*}) = \infty$,
thus finishing the proof.
It remains to prove~(\ref{eq:secondeq}).

\begin{description}
\item[Base case: $k=0$.]
By definition $n_0$ is such that $v_{n_0} < p_1$ and $v_{n_0+1}\geq p_1$.
Moreover, $p_1\in V$ by Observation~\ref{obs:banik}.
Hence, we have $v_{n_0+1} = p_1$, which establishes~(\ref{eq:secondeq}).
%

\item[Induction step: $k > 0$.]

We first prove that $y:=p_{k+1}$ satisfies all five conditions from
Definition~\ref{def:subproblem} for $\gain := \fgain(V,p_k,p_{k+1})$
and $\gainb := \fgainb(V,p_k,p_{k+1})$, and for $j := m_k$ and~$\delta' := \delta_k$.
This implies that $\fy(I_k) \geq p_{k+1}$.
After that we argue there is no larger value of $y$ satisfying all conditions,
thus finishing the proof.
\begin{enumerate}
\item $P_k$ wins at least $\gamma_k$ voters from $V_{n_k}$.
\vspace{2mm}

    $P_k$ wins $\gamma_k$ voters against~$Q_k$ by definition. Moreover, by the Claim above,
    $Q_k$ is an optimal strategy for \PlQ against $P_k$ on $V_{n_k}$.
    Hence, $P_k$ can win $\gamma_k$ voters from $V_{n_k}$.
\item $P_k$ induces an $\ell$-threshold $\tau_{\mathrm{opt}}$ with strictness $\delta_k$ on $V_{n_k}$.
\vspace{2mm}

    This is true by definition.
\item $\fspan(p_k,n_k,\gain,\gainb) = p_{k+1}$.
\vspace{2mm}

    Obviously $\fspan(p_k,n_k,\gain,\gainb) \geq p_{k+1}$. Now let $y := \fspan(p_k,n_k,\gain,\gainb)$
    and assume for a contradiction that $y > p_{k+1}$. In the next paragraph, we will argue that this implies that
    $P := \{p_1,\ldots,p_k,y,p_{k+2},\ldots,p_{k^*}\}$ realizes $\Gamma_{k^*,\ell^*}(V)$.
    This gives the desired contradiction since $P$ is then an optimal strategy that is
    lexicographically greater than the canonical strategy~$P_{\mathrm{opt}}$.

    Since the only difference between $P$ and $P_{\mathrm{opt}}$ is that $p_{k+1}$
    is moved to the right to the position~$y$,
    the intervals $(p_i,p_{i+1})$ only change for $i=k$ and for $i=k+1$.
    Thus the possible gains for \PlQ in all intervals stay the same,
    except possibly in $(p_k,p_{k+1})$ and $(p_{k+1},p_{k+2})$.
    Since $\gain = \fgain(V,p_k,p_{k+1})$ and $\gainb = \fgainb(V,p_k,p_{k+1})$
    and $y = \fspan(p_k,n_k,\gain,\gainb)$ by definition, the
    gains in $(p_k,p_{k+1})$ are the same as the gains in $(p_k,y)$.
    Finally, since $(y,p_{k+2})\subset (p_{k+1},p_{k+2})$, it is clear that
    \PlQ cannot win more voters in $(y,p_{k+2})$ against $P$
    than she can win in $(p_{k+1},p_{k+2})$ against~$P_{\mathrm{opt}}$
    by playing a single point inside these intervals. Similarly,
    she cannot win more voters by playing two points in $(y,p_{k+2})$,
    as compared to playing two points in $(p_{k+1},p_{k+2})$.
    Hence, $P$ wins at least the same number of voters as $P_{\mathrm{opt}}$
     and so $P$ realizes $\Gamma_{k^*,\ell^*}(V)$.
%
\item $P_{k-1}$ realizes $\fy(n_{k-1},k-1,\ell_k-j,\gamma_k-n_k+n_{k-1}+ \gain\gainplus_j\gainb,\delta_{k-1})$ for $j=\mmu_k$.
\vspace{2mm}

    By the induction hypothesis and since $\ell_{k-1} = \ell_k - j$ by definition,
    $P_{k-1}$ realizes $\fy(n_{k-1},k-1,\ell_k-j,\gamma_{k-1},\delta_{k-1})$.
    By the above Claim, $Q_{k-1}$ is optimal against $P_{k-1}$ on $V_{n_{k-1}}$ and $Q_k$ is
    optimal against $P_k$ on $V_{n_k}$. Hence, $\gamma_{k-1} = \gamma_k-n_k+n_{k-1}+ \gain\gainplus_j\gainb$, which proves the condition.
\item Let $\fmu(V_{n_{k-1}},P_{k-1},\ell_{k-1}) \coloneqq \langle \mmu'_0,\ldots,\mmu'_k \rangle$
      and $\fmu(V_{n_k},P_k,\ell_k) \coloneqq \langle \mmu_0,\ldots,\mmu_{k+1} \rangle$.
      Then $\mmu'_i = \mmu_i$ for all $0\leq i<k$.
\vspace{2mm}

      This is true by the above Claim.
\end{enumerate}
We still need to show that $p_{k+1}$ is the maximum value in $(v_n,v_{n+1}]$ satisfying the above mentioned conditions.
This is obvious for $k=k^*$. For $k<k^*$,
with a similar reasoning to the proof of the third condition,
we can show that $p_{k+1}$ is the maximum value in $(v_n,v_{n+1}]$ satisfying
the conditions, as otherwise $P_{\mathrm{opt}}$ cannot be the canonical
strategy realizing $\Gamma_{k^*,\ell^*}(V)$. Namely, we assume $\vy \in (v_n,v_{n+1}]$
with $\vy>p_{k+1}$ is the solution to $I_k$ realized by $P'_k := \{p'_1,\ldots,p'_k,\vy\}$,
and then we argue that $P := \{p'_1,\ldots,p'_k,\vy,p_{k+2},\ldots,p_{k^*}\}$
is an optimal solution, thus obtaining a contradiction since $P$
is lexicographically larger than $P_{\mathrm{opt}}$. It remains to
prove that~$P$ is optimal.

We show that any other strategy $Q$
against $P$ on $V$ cannot win more voters than $Q_{\mathrm{opt}}$ wins against $P_{\mathrm{opt}}$, which shows $P$ realizes $\Gamma_{k^*,\ell^*}(V)$
and is therefore optimal.
Let $Q'_k$ be the canonical strategy of \PlQ with $\ell_k$ points against $P'_k$ on $V_{n_k}$. As $P'_k$ wins at least $\gamma_k$ voters and by definition of $\gamma_k$, $P_k$
wins exactly $\gamma_k$ voters, the number of voters $Q'_k$ wins against $P'_k$ cannot be more than the number of voters $Q_k$, which is part of $Q_{\mathrm{opt}}$, wins against $P_k$.

Let $\gain\coloneqq\fgain(V,p_{k+1},p_{k+2})$, and $\gainb\coloneqq\fgainb(V,p_{k+1},p_{k+2})$.
Any strategy $Q$ against $P$ must select its gains from the gain sequence $\ftau(V,P)$. But, all the gains in $\ftau(V,P)$ are either the same gains
chosen by $Q'_k$ or $Q_{\mathrm{opt}}$ of a value at least $\tau_{\mathrm{opt}}$, or they have a value of at most $\tau_{\mathrm{opt}}$ because of the threshold $\tau_{\mathrm{opt}}$.
The only exceptions are the gains in interval $(y,p_{k+2})$
where by moving $p_{k+1}$ to the right $\gain$ does not increase, but $\gainb$ might increase.
But, even in that case when $\gainb$ increases and its value gets bigger than $\tau_{\mathrm{opt}}$ and it is chosen by $Q$, it means $\gain$
is chosen by both $Q_{\mathrm{opt}}$ and $Q$. Hence, any extra voters won by $Q$ through selecting the gain $\gainb$ are stolen from $\gain$ which means
the number of voters won by $Q$ does not increase.

\end{description}
\end{proof}

\paragraph{{\rm\emph{Remark.}}}
Usually in dynamic programming, subproblems have a clean non-recursive definition---the recursion only comes in when a recursive formula is given to compute the value of an optimal solution. Our approach is more complicated: Definition~\ref{def:subproblem} above gives a recursive but ``non-constructive' subproblem definition (and Lemma~\ref{lem:computeGamma} shows how to use it) and Lemma~\ref{lem:basis-for-dp} below will then give a different recursive formula to actually compute the solutions to the subproblems.



\subsection{Computing solutions to subproblems}
The solution to an elementary subproblem follows fairly easily from the definitions, and can be computed in constant time.

\begin{lemma}
\label{lem:elementary}
Assuming the voter set~$V$ is given in sorted order, we can find the solution to an
elementary subproblem in constant time.
\end{lemma}
\begin{proof}
Consider an elementary subproblem $I = \langle n,0,\ell,\gamma,\delta\rangle$.
If $\ell=0$ then $\fy(I)=v_{n+1}$ if and only if the following two conditions hold:
(i) $\gamma \leq n$ and (ii) $\tau > n $ or $(\tau=n \mbox{ and } \delta = \loose)$.
If $\ell=1$ then $\fy(I)=v_{n+1}$ if and only if
(i) $\gamma = 0$ and (ii) $\tau < n$ or $(\tau=n \mbox{ and } \delta = \strict)$.
Otherwise $\fy(I)=-\infty$.
\end{proof}

By definition, in order to obtain a strategy $P$ realizing the solution to a non-elementary subproblem $I=\langle n,k,\ell,\gamma,\delta\rangle$ of size $k$, we need a solution to a smaller subproblem $I'=\langle n',k-1,\ell',\gamma',\delta'\rangle$ of size $k-1$ and add one point $\vy\in(v_n,v_{n+1}]$ to the strategy $P'=\{p_1,\ldots,p_k\}$ realizing $I'$. Thus by adding $y$, we extend the solution to $I'$ to get a solution to $I$.
To find the ``right'' subproblem $I'$, we guess some values for $n'$, $\gain$, $\gainb$, $j\in\{0,1,2\}$, and $\delta'\in\{\strict,\loose\}$; these values are enough to specify $I'$.
We note that there are just a polynomial number of cases and therefore a polynomial number of values for the value $\vy\in(v_n,v_{n+1}]$ which we want to maximize. Namely, there are $O(n)$ choices for the values $n'$, $\gain$, and $\gainb$, three choices for $j$, and two choices for $\delta'$. This makes $O(n^3)$ different cases to be considered for each subproblem $I$, in total. However, not all those subproblems can be extended to $I$. In the following definition, we list the triples $(\delta',j,\delta)$ that provide all the valid combinations that guarantee the extendibility of $I'$ to $I$.

Let $\gain$ and $\gainb$ be the $\ngain$-gain and $\ngainb$-gain of the interval $(p_k,\vy)$ in a strategy $P=\{p_1,\ldots,p_k,\vy\}$ with threshold $\tau$.
We define the following sets of triples depending on the relationship between $a$, $b$ and $\tau$:
\begin{align*}
\mkern-56mu 
\Delta&(\tau,a,b) \coloneqq \\
& 
\begin{cases}
\{(\loose,2,\loose),(\strict,2,\strict)\}
& \mbox{if } \gain > \tau \land \gainb > \tau \\
\{(\loose,1,\loose),(\strict,1,\loose),(\strict,2,\strict)\}
& \mbox{if } \gain > \tau \land \gainb = \tau \\
\{(\loose,1,\loose),(\strict,1,\strict)\}
& \mbox{if } \gain > \tau \land \gainb < \tau \\
\{(\loose,0,\loose),(\strict,0,\loose),(\strict,1,\loose),(\strict,2,\strict)\}
& \mbox{if } \gain = \tau \land \gainb = \tau \\
\{(\loose,0,\loose),(\strict,0,\loose),(\strict,1,\strict)\}
& \mbox{if } \gain = \tau \land \gainb < \tau \\
\{(\loose,0,\loose),(\strict,0,\strict)\}
& \mbox{if } \gain < \tau \land \gainb < \tau.
\end{cases}
\end{align*}
\begin{lemma}
\label{lem:1d:sufficient}
Let $P'=\{p_1,\ldots,p_k\}$ and $P\coloneqq P'\cup\{\vy\}$, be two reasonable strategies on $V_{n'}$ and $V_n$, where $n'=\argmax_{1\leq i \leq n^*} v_i < p_k$, $n=\argmax_{1\leq i \leq n^*} v_i < \vy$, and $\vy\in(v_n,v_{n+1}]$. Let $\gain = \fgain(V_n,p_k,y)$ and $\gainb = \fgainb(V_n,p_k,y)$,
and assume $\tau>0$ is an $(\ell-j)$-threshold of strictness $\delta'$ for \PlQ induced by $P'$ on $V_{n'}$, where $j\in\{0,1,2\}$. Then, there exists a triple $(\delta',\jmu,\delta)\in\Delta(\tau,\gain,\gainb)$ if and only if
\myenumerate{
\item \label{lem:1d:sufficient:cond1} $P$ induces an $\ell$-threshold $\tau$ with strictness $\delta$ on $V_n$,
\item \label{lem:1d:sufficient:cond2} Let $\fmu(V_{n'},P',\ell-j) \coloneqq \langle \mmu'_0,\ldots,\mmu'_k \rangle$ and $\fmu(V_n,P,\ell) \coloneqq \langle \mmu_0,\ldots,\mmu_{k+1} \rangle$. Then
   $\mmu'_i = \mmu_i$ for all $0\leq i<k$.
}
Moreover, this triple is unique if it exists.
\end{lemma}
\begin{proof}
We prove the lemma for the case $\tau=\gain=\gainb$. The proof for the other cases is similar.

Let $Q'$ and $Q$ be canonical strategies associated with $\fmu(V_{n'},P',\ell-j)$ and $\fmu(V_n,P,\ell)$, respectively, and
let $\ftau(V_{n'},P') = \langle \tau'_1,\ldots,\tau'_{2k+2} \rangle$.
We note that $\ftau(V_n,P)$ can be constructed by inserting $\gain$ and $\gainb$ after the last $\tau$ value in $\ftau(V_{n'},P')$. There are six cases,
\begin{itemize}
\item $\delta'=\loose$ and $j=0$.

In this case $(\loose,0,\loose)$ is the only corresponding triple in $\Delta(\tau,\gain,\gainb)$.
As $\delta'=\loose$, we have $\tau'_{\ell-j+1}=\tau$, and therefore $\gain$ and $\gainb$ are not selected by $Q$. Therefore $Q'=Q$ and the strictness status also does not change. Hence both conditions are satisfied and $\delta=\loose$.

\item $\delta'=\strict$ and $j=0$.

In this case $(\strict,0,\loose)$ is the only corresponding triple in $\Delta(\tau,\gain,\gainb)$.
As $\delta'=\strict$, we have $\tau'_{\ell-j+1}<\tau$, and therefore $\gain$ and $\gainb$ are not selected by $Q$; however, the strictness of the $\ell$-threshold $\tau$ for $P$ is different. Hence both conditions are satisfied and $\delta=\loose$.
\item $\delta'=\loose$ and $j=1$.

As $\delta'=\loose$, we have $\tau'_{\ell-j+1}=\tau$, and therefore $\gain$ and $\gainb$ are not selected by $Q$. However, the gain $\tau'_{\ell-j+1}$ is selected by $Q$ which violates the condition \ref{lem:1d:sufficient:cond2}. We also note that there is no corresponding triple in $\Delta(\tau,\gain,\gainb)$ for this case.
\item $\delta'=\strict$ and $j=1$.

In this case $(\strict,1,\loose)$ is the only corresponding triple in $\Delta(\tau,\gain,\gainb)$.
As $\delta'=\strict$, we have $\tau'_{\ell-j+1}<\tau$, and therefore $\gain$ is selected by $Q$ and $\gainb$ is not selected, and the strictness of the $\ell$-threshold $\tau$ for $P$ is different. Hence both conditions are satisfied and $\delta=\loose$.

\item $\delta'=\strict$ and $j=2$.

In this case $(\strict,2,\strict)$ is the only corresponding triple in $\Delta(\tau,\gain,\gainb)$.
As $\delta'=\strict$, we have $\tau'_{\ell-j+1}<\tau$, and therefore both $\gain$ and $\gainb$ are selected by $Q$, and the strictness of the $\ell$-threshold $\tau$ for $P$ is the same. Hence both conditions are satisfied and $\delta=\strict$.
\item $\delta'=\loose$ and $j=2$.

As $\delta'=\loose$, we have $\tau'_{\ell-j+1}=\tau$, and therefore just $\gain$ is selected by $Q$, while $\gainb$ is not selected. However, the gain $\tau'_{\ell-j+1}$ is selected by $Q$ which violates the condition \ref{lem:1d:sufficient:cond2}. We also note that there is no corresponding triple in $\Delta(\tau,\gain,\gainb)$ for this case.
\end{itemize}
\end{proof}

%
%
\begin{lemma}
\label{lem:basis-for-dp}
For a non-elementary subproblem $I=\langle n,k,\ell,\gamma,\delta\rangle$, we have
\begin{multline*}
\fy(n,k,\ell,\gamma,\delta) = \max_{0\leq n' < n}\max_{0\leq\gain\leq n}\max_{0\leq\gainb\leq n}\max_{(\delta',j,\delta)\in\Delta(\tau,\gain,\gainb)}\\
\fspan\big(\fy(n',k-1,\ell-\jmu,\gamma-n+n'+\gain\gainplus_\jmu\gainb,\delta'),n,\gain,\gainb\big).
\end{multline*}
\end{lemma}
\begin{proof}
As stated earlier, when one of the conditions $\ell < 2(k+1)$ or $k \leq \|V_n\|$ is violated, the left hand side evaluates to $-\infty$. When this happens, then the corresponding condition for each subproblem $\fy(n',k-1,\ell-\jmu,\gamma-n+n'+\gain\gainplus_\jmu\gainb,\delta')$ is also violated, i.e. $(\ell-\jmu) < 2((k-1)+1)$ or $(k-1) \leq \|V_{n'}\|$. Hence, each of these subproblems and the right hand side evaluates to $-\infty$, too. When $\ell < 2(k+1)$ and $k \leq \|V_n\|$ are both satisfied, we prove the lemma as follows.

Let $\vy$ be the solution to $I$ and let this solution be realized by the strategy $P\coloneqq P'\cup\{\vy\}$ for some $P'=\{p_1,\ldots,p_k\}$. By definition, $P'$ realizes $\fy(n',k-1,\ell-\jmu,\gamma-n+n'+\gain\gainplus_\jmu\gainb,\delta')$  for some $\delta'\in\{\loose,\strict\}$ and $\jmu \in \{0,1,2\}$ where $\gain = \fgain(V_n,p_k,\vy)$, $\gainb = \fgainb(V_n,p_k,\vy)$, and $n'=\argmax_i v_i < p_k$. By the second and fifth conditions of a subproblem and Lemma~\ref{lem:1d:sufficient}, we have $(\delta',\jmu,\delta)\in\Delta(\tau,\gain,\gainb)$, which shows
\begin{multline*}
\fy(n,k,\ell,\gamma,\delta) \leq \max_{0\leq n' < n}\max_{0\leq\gain\leq n}\max_{0\leq\gainb\leq \gain}\max_{(\delta',j,\delta)\in\Delta(\tau,\gain,\gainb)}\\
\fspan(\fy(n',k-1,\ell-\jmu,\gamma-n+n'+\gain\gainplus_\jmu\gainb,\delta'),n,\gain,\gainb).
\end{multline*}

Conversely, consider a value $\vy=\fspan(\vy',n,\gain,\gainb)>-\infty$, where $\vy'=\fy(n',k-1,\ell-\jmu,\gamma-n'+n''+\gain\gainplus_\jmu\gainb,\delta')$ is realized by $P'=\{p_1,\ldots,p_k\}$ for some values $n'$,$\gain$,$\gainb$,$\delta'$, and $j$ such that $(\delta',j,\delta)\in\Delta(\tau,\gain,\gainb)$. Let $P\coloneqq P'\cup\{\vy\}$. It is not hard to verify that the conditions for $\vy$ to be a solution to $I$ realized by $P$ are satisfied by these values. Indeed, we have:
\begin{enumerate}
\item $P$ wins at least $\gamma$ voters from $V_{n}$.

By the second and the fifth conditions and the fact that $\tau>0$, we have $\fmu_k=j$ which shows \PlP can win at least $\gamma$ voters.

\item $P$ induces an $\ell$-threshold $\tau$ with strictness $\delta$ on $V_n$.

This condition is satisfied by Condition \ref{lem:1d:sufficient:cond1} of Lemma~\ref{lem:1d:sufficient} and by the uniqueness of $\delta$.

\item $\fspan(p_k,n,\gain,\gainb) = \vy$.

This condition is satisfied by the assumption.

\item $P'$ realizes $\fy(n',k-1,\ell-\jmu,\gamma-n+n'+\gain\gainplus_\jmu\gainb,\delta')$.

This condition is satisfied by the assumption.

\item Let $\fmu(V_{n'},P',\ell-j) = \langle \mmu'_0,\ldots,\mmu'_k \rangle$ and $\fmu(V_n,P,\ell) = \langle \mmu_0,\ldots,\mmu_{k+1} \rangle$. Then
   $\mmu'_i = \mmu_i$ for all $0\leq i<k$.

This condition is satisfied by Condition \ref{lem:1d:sufficient:cond2} of Lemma~\ref{lem:1d:sufficient}.
\end{enumerate}

As all the conditions are satisfied, we have
\begin{multline*}
\fy(n,k,\ell,\gamma,\delta) \geq \max_{0\leq n' < n}\max_{0\leq\gain\leq n}\max_{0\leq\gainb\leq \gain}\max_{(\delta',j,\delta)\in\Delta(\tau,\gain,\gainb)}\\
\fspan(\fy(n',k-1,\ell-\jmu,\gamma-n+n'+\gain\gainplus_\jmu\gainb,\delta'),n,\gain,\gainb),
\end{multline*}
which completes the proof.
\end{proof}

If we can compute the $\fspan$ function efficiently, we can compute all the solutions by dynamic programming and solve the problem.
However, a solution based on a trivial dynamic-programming algorithm will be of running time complexity $\lfloor n^*/\ell^*\rfloor \,\cdot\, O(k^* \ell^* (n^*)^2 )\,\cdot\, O((n^*)^3 f(n^*)) = O(k^* (n^*)^6 f(n^*))$ where $\lfloor n^*/\ell^*\rfloor$ is the total number of choices for the threshold, $O(k^* \ell^* (n^*)^2 )$ is the number of subproblems for each threshold,
and $O((n^*)^3 f(n^*))$ is the time needed to solve each subproblem where $f(n)$ is the time needed to compute the $\fspan(x,n,a,b)$ function. This algorithm is quite slow. More importantly
it is not easy to compute the $\fspan$ function. In the following, we introduce some new concepts to compute the $\fspan$ function and also get a better running time.

\subsection{Computing the \texorpdfstring{$\fspan$}{span} function using gain maps}
\begin{figure}
  \centering
  \includegraphics{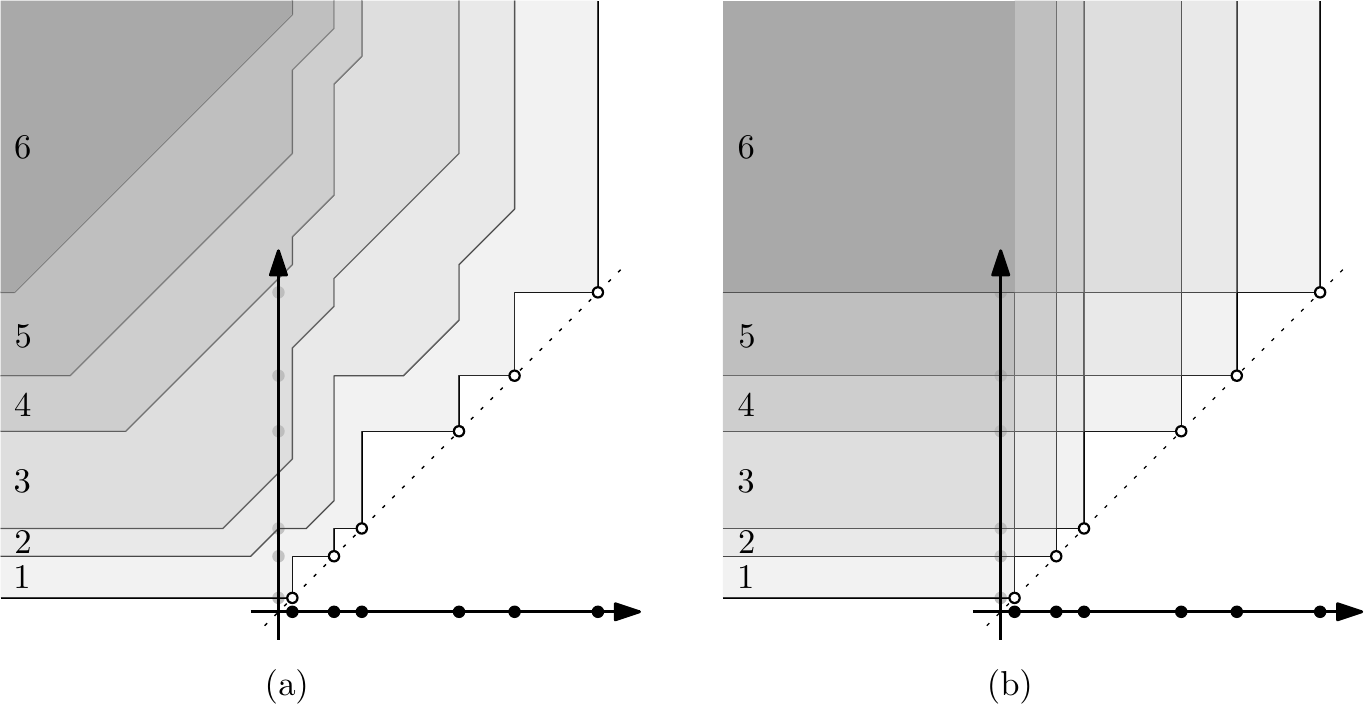}\\
  \caption{{\bf a)} $A$-map of $V=\{1,4,6,13,17,23\}$ with the corresponding $\ngain$-gain for each region. {\bf b)}~$B$-map of $V$ with the corresponding $\ngainb$-gain for each region.}
  \label{fig:gainmap}
\end{figure}
Before we give the algorithm we introduce the \emph{gain map}, which we need
to compute the $\fspan$ function.
Consider an arbitrary strategy $P$ of \PlP on $V$, and recall that such a strategy
induces open intervals of the form~$(p_i,p_{i+1})$ where \PlQ can place her points.
We can represent any possible interval~$(x,y)$ that may arise
in this manner as a point $(x,y)$ in the plane.
Thus the locus of all possible intervals is the region $R := \{ (x,y) : x < y \}$.
We will define two subdivisions of this region, the $A$-map and the $B$-map,
and the gain map will then be the overlay of the $A$-map and the $B$-map.
\medskip

The $A$-map is the subdivision of $R$ into regions $A^t$ and $B^t$, for $0\leq t \leq n^*$, defined
as follows:
\begin{align*}
A^t &:= \{ (x,y) : \fgain(V,x,y) = t \}  \mbox{\hspace*{0mm} and \hspace*{0mm} } \\ 
B^t &:= \{ (x,y) : \fgain(V,x,y) + \fgainb(V,x,y)=t \}. 
\end{align*}
In other words, $A^t$ is the locus of all intervals $(x,y)$ such that,
if $(x,y)$ is an interval induced by $P$, then \PlQ can win $t$ voters
 (but no more than $t$) from $V\cap (x,y)$ by placing a single point in~$(x,y)$.
To construct the $A$-map, let $A^{\geq t}$ denote the locus of all intervals $(x,y)$ such that
$\fgain(V,x,y) \geq t$.
Note that $A^t = A^{\geq t} \setminus A^{\geq t+1}$. For $1\leq i \leq n^*-t+1$,
let $V_i^t := \{ v_i,\ldots, v_{i+t-1} \}$ and define
\begin{align*}
\mkern-26mu
A^{\geq t}_i := \{ (x,y) : &V_i^t \subset (x,y) \mbox{ and}\\
&\mbox{\PlQ can win all voters in $V_i^t$ by placing a single point in $(x,y)$} \}. 
\end{align*}
Then we have
$
A^{\geq t}_i = \{ (x,y) : x < v_i \mbox{ and } y > v_{i+t-1} \mbox{ and } y > x + 2(v_{i+t-1}-v_i) \}.
$
Here the conditions $x < v_i$ and $y > v_{i+t-1}$ are needed to guarantee that $V_i^t \subset (x,y)$.
The condition $y > x + 2(v_{i+t-1}-v_i)$ implies that $V_i^t$ can be covered with an interval
of length $(y-x)/2$, which is necessary and sufficient for \PlQ to be able to win all these voters.
Note that each region~$A^{\geq t}_i$ is the intersection of three halfplanes, bounded
by a vertical, a horizontal and a diagonal line, respectively.

Since \PlQ can win at least $t$ voters in inside $(x,y)$ with a single point
if she can win at least $t$ consecutive voters with a single point, we have
$A^{\geq t} = \bigcup_{i=1}^{n^*-t+1}A^{\geq t}_i$.
Thus $A^{\geq t}$ is a polygonal region, bounded from below and from the right by a
a polyline consisting of horizontal, vertical, and diagonal segments,
and the regions $A^{t}$ are sandwiched between such polylines;
see Figure~\ref{fig:gainmap}a. We call the polylines that form
the boundary between consecutive regions $A^{t}$ \emph{boundary polylines}.

The $B$-map can be constructed in a similar, but easier manner. Indeed,
$B^t$ is the locus of all intervals such that \PlQ can win $t$ voters
(but no more) from $V\cap (x,y)$, and this is the case if and only if
$|V\cap (x,y)| = t$. Hence, $B^t$ is the union of the rectangular regions
$[v_i,v_{i+1}) \times (v_{i+t},v_{i+t+1}]$ (intersected with $R$), for $0\leq i\leq n^*-t$,
where $v_0\coloneqq -\infty$ and $v_{n^*+1} \coloneqq \infty$,
as shown in Figure~\ref{fig:gainmap}b. 
%
\medskip

As mentioned, by overlaying the $A$- and $B$-map, we get the \emph{gain map}.
For any given region on this map, the intervals corresponding to points inside this region
have equal $\ngain$-gain and equal $\ngainb$-gain.
\begin{lemma}
The complexity of the gain-map is $O((n^*)^2)$.
\end{lemma}
\begin{proof}
The boundary polylines in the $A$-map are $xy$-monotone and
comprised of vertical, horizontal, and diagonal lines.
The $B$-map is essentially a grid of size $O((n^*)^2)$ defined by the lines $x=v_i$ and $y=v_i$,
for $1\leq i\leq n^*$. Since each of these lines intersects any $xy$-monotone polyline
at most once---in a point or in a vertical segment---the
complexity of the gain map is also $O((n^*)^2)$.
\end{proof}
Using the gain map, we can compute the values $\fspan(x_0,n,\gain,\gainb)$
for a given $x_0\in\Reals$ and for all triples $n,\gain,\gainb$ satisfying $1\leq n \leq n^*$, and
$0\leq \gain \leq n^*$ and $0\leq \gainb \leq n^*$, as follows.
First, we compute the intersection points of the vertical line~$x=x_0$ with
(the edges of) the gain map, sorted by increasing $y$-coordinates.
(If this line intersects the gain map in a vertical segment, we take the topmost endpoint of the segment.)
Let $(x_0,y_1),\ldots,(x_0,y_z)$ denote this sorted sequence of
intersection points, where $z\leq 2n^*$ denotes the number of intersections.
Let $\gain_i$ and $\gainb_i$ denote the $\ngain$-gain and $\ngainb$-gain
of the interval corresponding to the point $(x_0,y_i)$, and let $\gain_{z+1}$ and $\gainb_{z+1}$
denote the $\ngain$-gain and $\ngainb$-gain of the unbounded region intersected by the line~$x=x_0$.
Define $n_i=\argmax_{n} v_{n} < y_i$. Then we have
\begin{equation} \label{eq:span-sweep}
\fspan(x_0,n,\gain,\gainb) =
\begin{cases}
y_i &  \mbox{if $a=a_i$ and $b=b_i$, and $n=n_i$, for some $1\leq i \leq z$} \\
+\infty & \mbox{if $a=a_{z+1}$ and $b=b_{z+1}$ and $n=n^*$} \\
-\infty & \mbox{for all other triples $n,a,b$} \\
\end{cases}
\end{equation}

Our algorithm presented below moves a sweep line from left to right over the gain map.
During the sweep we maintain the intersections of the sweep line with the
gain map. It will be convenient to maintain the intersections with the
$A$-map and the $B$-map separately. We will do so using two sequences,
$\fa(x_0)$ and $\fb(x_0)$.
\myitemize{
\item The sequence $\fa(x_0)$ is the sequence of all diagonal or horizontal edges
      in the $A$-map that are intersected by the line~$x=x_0$, ordered from
      bottom to top along the line. (More precisely, the sequence contains
      (at most) one edge for any boundary polyline. When the sweep line reaches the endpoint of such an
      edge, the edge will be removed and it will be replaced by the next non-vertical
      edge of that boundary polyline, if it exists.)
\item The sequence $\fb(x_0)$ is the sequence of the $y$-coordinates of the horizontal segments in the $B$-map intersected by the line $x=x_0$, ordered from
      bottom to top along the line.
}
The number of intersections of the $A$-map, and also of the $B$-map,
with the line~$x=x_0$ is equal to $n^*-n_0+1$, where $n_0=\argmin_{n} v_{n} > x_0$.
Hence, the sequences $\fa(x_0)$ and $\fb(x_0)$ have length~$n^*-n_0+1 \leq n^*+1$.

If we have the sequences $\fa(x_0)$ and $\fb(x_0)$ available then,
using Equation~(\ref{eq:span-sweep}), we can easily find all triples
$n,\gain,\gainb$ such that $\fspan(x_0,n,\gain,\gainb) \neq -\infty$
(and the corresponding $y$-values)
by iterating over the two sequences.
We summarize the results of this section in the following observation.
\begin{observation}
\label{obs:span}
If we know the sequences $\fa(x_0)$ and $\fb(x_0)$ then we can compute all the values $\fspan(x_0,n,\gain,\gainb)$, with $1\leq n\leq n^*$ and $0\leq a,b\leq n$,
that are not equal to $-\infty$ in $O(n^*)$ time in total.
\end{observation}
This observation, together with Lemma~\ref{lem:basis-for-dp} forms the basis
of our dynamic-programming algorithm.


\subsection{The sweep-line based dynamic-programming algorithm}

Usually in a dynamic-programming algorithm, the value of a subproblem is computed
by looking up the values of certain smaller subproblems. In our case it is hard
to determine which smaller subproblems we need, so we take the opposite approach:
whenever we have computed the value of a subproblem we determine which other subproblems
can use this value, and we update their solutions if necessary.
To this end we will use a sweep-line approach, moving a vertical line from left to right
over the gain map. We will maintain a table~$X$, indexed by subproblems, such that
when the sweep line is at position~$x_0$, then $X[I]$ holds the best solution
known so far for subproblem~$I$, where the effect of all the subproblems with
solution smaller than $x_0$ have been taken into account. When our sweep
reaches a subproblem~$I'$, then we check which later
subproblems~$I$ can use $I'$ in their solution, and we update
the solutions to these subproblems. 

%
Recall the algorithm works with a fixed threshold value $\tau\in\{1,\ldots,\lfloor n^*/\ell^*\rfloor\}$ and that its goal is to compute the values $\fy(n^*,k^*,\ell^*,\gamma,\delta)$ for all
$0\leq \gamma\leq n^*$ and $\delta\in\{\strict,\loose\}$.
Our algorithm maintains the following data structures.
\myitemize{
\item $A[0..n^*]$ is an array that stores the sequence $\fa(x_0)$, where $x_0$ is the current
      position of the sweep line and $A[i]$ contains the $i$-th element in the sequence.
      When the $i$-th element does not exist then $A[i]=\NA$.
\item Similarly, $B[0..n^*]$ is an array that stores the sequence $\fb(x_0)$.
\item $X$: This is a table with an entry for each
      subproblem~$I=\langle n,k,\ell,\gamma,\delta \rangle$ with
      $0\leq n\leq n^*$, and $0\leq k\leq k^*$ and $0\leq \ell\leq \ell^*$,
      and $0\leq \gamma\leq n^*$ and $\delta\in\{\strict,\loose\}$.
      When the sweep line is at position~$x_0$, then $X[I]$ holds the best solution
      known so far for subproblem~$I$, where the effect of all the subproblems with
      solution smaller than $x_0$ have been taken into account. More precisely, in the
      right-hand side of the equation in Lemma~\ref{lem:basis-for-dp} we have
      taken the maximum value over all subproblems $I' = \langle n',k-1,\ell-\jmu,\gamma-n+n'+\gain\gainplus_\jmu\gainb,\delta' \rangle$
      with $\fy(I') < x_0$.
      In the beginning of the algorithm the entries for elementary subproblems are
      computed using Lemma~\ref{lem:elementary} and all other entries have value~$-\infty$.
\item $E$: This is the event queue, which will contain four types of events, as explained below.
}
The event queue $E$ is a min-priority queue on the $x$-value of the events.
There are four types of events, as listed next, and when events have the
same $x$-value then the first event type (in the list below) has higher priority,
that is, will be handled first. When two events of the same type have equal
$x$-value then their order is arbitrary. Note that events with the
same $x$-value are not degenerate cases---this is inherent to the
structure of the algorithms, as many events take place at $x$-coordinates corresponding to voters.
\begin{description}
\item[$A$-map events, denoted by $e_A(\gain,s,s')$:] At an $A$-map event, the edge~$s$
     of the $A$-map ends---thus the $x$-value of an $A$-map event is the $x$-coordinate
     of the right endpoint of $s$---and the array~$A$ must be updated
     by replacing it with the edge $s'$.
     Here $s'$ is the next non-vertical edge along the
     boundary polyline that $s$ is part of, where $s' = \NA$ if $s$ is the last
     non-vertical edge of the boundary polyline. The value~$\gain$ indicates that
     the edge $s$ is on the boundary polyline between $A^a$ and $A^{a+1}$.
     In other words $s$ (and $s'$, if it exist) are the $a$-th intersection
     point, $0\leq a < n^*$, with the $A$-map along the current sweep line,
     and so we must update the entry $A[a]$ by setting $A[a] \leftarrow s'$.
\item[$B$-map events, denoted by $e_B(v_n)$:] At a $B$-map event, a horizontal edge of
      the $B$-map ends. This happens when the sweep line reaches a voter $v_n$---that is,
      when $x_0=x_n$---and so the $x$-value of this event is~$v_n$.
      The bottommost intersection of he sweep line with the $B$-map
      now disappears (see Figure~\ref{fig:gainmap}b), and so
      we must update $B$ by shifting all other intersection points one position
      down in $B$ and setting $B[n^*-n]\leftarrow\NA$.
\item[Subproblem events, denoted by $e_X(n',k',\ell',\gamma',\delta')$:]
    At a subproblem event the solution to the subproblem given by
    $I' = \langle n',k',\ell',\gamma',\delta' \rangle$ is known and
    the $x$-value of this event is equal to $\fy(I')$.
    Handling the subproblem event for $I'$ entails deciding which later
    subproblems~$I$ can use $I'$ in their solution and how they can use it,
    using the sets $\Delta(\tau,\gain,\gainb)$, and updating
    the solutions to these subproblems.

    In the beginning of the algorithm all the events associated with elementary subproblems are known.
    The events associated with non-elementary subproblems are added to the event queue
    when handling an update event $e_E(v_n)$, as discussed next.
\item[Update events, denoted by $e_E(v_n)$:] At the update event happening at $x$-value~$v_n$,
    all subproblem events of size~$n$ are added to the event queue~$E$. These are simply the
    subproblems $\langle n,k,\ell,\gamma,\delta\rangle$ for all $k,\ell,\gamma\in\{0,\ldots,n\}$ and
    $\delta\in\{\strict,\loose\}$. The reason we could not add them at the start of the algorithm
    was that the $x$-value of such a subproblem~$I$ was now known yet. However, when we
    reach $v_n$ then $\fy(I)$ is determined, so we can add the event to $E$ with $\fy(I)$
    as its $x$-value.
\end{description}
The pseudocode below summarizes the algorithm.
\begin{algorithm}
\LinesNumbered
\DontPrintSemicolon
\SetArgSty{}
\SetKwIF{If}{ElseIf}{Else}{if}{}{else if}{else}{end if}
        \For {$i \leftarrow 0$ \KwTo $n^*-1$} {
             $A[i]\leftarrow (v_i,v_{i+1})-(v_{i+1},v_{i+1})$; \hspace*{3mm} $B[i]\leftarrow v_{i+1}$ \Comment*[r]{define $v_0\coloneqq v_1-1$}
        }
     $A[n^*]\leftarrow \NA$; \hspace*{3mm} $B[n^*]\leftarrow \NA$\;
	 Initialize $X$ by the solutions to elementary subproblems \;
        Initialize $E$ by all map events, update events, and elementary subproblem events \;
        \While {$E$ is not empty} { \label{alg:event-loop}
            $e \leftarrow \mathrm{extractMin}(E)$; \hspace*{3mm} $x_0 \leftarrow x$-value of $e$\;
            \Switch{$e$}{
                \Case{$e_A(\gain,s,s')$}{
			 $A[\gain] \leftarrow s'$ \;
                }
                \Case{$e_E(v_n)$} {
			  $B[n^*-n] \leftarrow \NA$\;
		        \For {$i \leftarrow 0$ \KwTo $n^*-n-1$} {
             		$B[i]\leftarrow v_{n+i+1}$\;
        		}
                }
                \Case{$e_X(n',k',\ell',\gamma',\delta')$} {
                    \For {all $\fspan(x_0,n,\gain,\gainb)=y$ where $y\neq -\infty$} {
				\For {all $(\delta',j,\delta) \in \Delta(\tau,\gain,\gainb)$} {
				$I \leftarrow \langle n,k'+1,\ell'+j,\gamma'+n-n'-\fg^j(\gain,\gainb),\delta \rangle$ \;
                           	$X(I) \leftarrow \max(X(I),y)$ \;
				}
                    }
                }
                \Case{$e_E(v_n)$} {
                    Add all the events for subproblems of size $n$ to $E$, as explained above \;
                }
            }
        }
    \caption{$\ComputeSolutions(\tau,V,k^*,\ell^*)$}\label{alg:sweepline}
\end{algorithm}
%
\begin{lemma}
\label{lem:alg-correct}
Algorithm~\ref{alg:sweepline} correctly computes the solutions for subproblems
$\langle n,k,\ell,\gamma,\delta\rangle$ for the given value $\tau$,
for all $n,k,\ell,\gamma,\delta$ with $0\leq n \leq n^*$, and $0\leq k,\ell,\gamma \leq n$, and $\delta\in\{\strict,\loose\}$, and $\ell < 2(k+1)$.
The running time of the algorithm is $O(k^*\ell^*(n^*)^3)$.
\end{lemma}
\begin{proof}
We handle the $A$-map and $B$-map events before a subproblem event so that $A$ and $B$ data structures are up-to-date when we want to compute the $\fspan$ function on handling a subproblem event.
We also handle a subproblem event before an update event so that when we want to add a new subproblem event to the event queue on handling an update event, its entry in table $X$ has the correct value.
The correctness of the algorithm now follows from the discussion and lemmas above.

The running time is dominated by the handling of the subproblem events.
By Observation~\ref{obs:span}, the algorithm handles each subproblem in $O(n^*)$ time,
plus $O(\log n^*)$ for operations on the event queue,
and there are $O(k^*\ell^*(n^*)^2)$ subproblems. Hence,  the total
running time is~$O(k^*\ell^*(n^*)^3)$.
\end{proof}
By Lemmas~\ref{lem:computeGamma} and~\ref{lem:alg-correct}, the algorithm described
at the beginning of Section~\ref{sec:subproblem} computes $\Gamma_{k^*,\ell^*}(V)$ correctly.
Since this algorithm calls \ComputeSolutions $\lfloor n^*/{\ell^*} \rfloor$ times in Step~\ref{step1},
we obtain the following theorem.
\begin{theorem}
There exists an algorithm that computes $\Gamma_{k^*,\ell^*}(V)$, and thus solves
the one-dimensional case of the one-round discrete Voronoi game, in time~$O(k^*(n^*)^4)$.
\end{theorem}
\paragraph{{\rm\emph{Remark.}}}
We can also solve the one-dimensional case of the one-round discrete Voronoi game when voters are weighted, i.e.
each voter $v\in V$ has an associated weight $\omega(v)$ and the players try to maximize the total weight of the voters they win.
In this case, the $\ngain$-gain and $\ngainb$-gain of an interval is defined as the total weight of voters the second player
can win in that interval by placing one point and two points, respectively. The number of possible thresholds is not an integer in
range $[0,n^*]$, but the sum of any sequence of consecutive voters define a threshold, which makes a total of $O((n^*)^2)$ different thresholds.
The gain map also becomes more complex and in the algorithm we need to spend $O((n^*)^2)$ time (instead of $O(n^*)$) to handle
each subproblem event, which results in an algorithm with running time $O(k^*\ell^*(n^*)^5)$.

\section{\texorpdfstring{$\Sigma_2^P$}{Sigma-2}-Hardness for \texorpdfstring{$d\geq 2$}{d>1}}
In this section we prove that the one-round discrete Voronoi game is $\Sigma_2^P$-hard
in $\Reals^2$, which implies hardness for $d>2$ as well.
To prove this, it suffices to show that deciding if \PlQ has a winning strategy against every possible strategy of \PlP is $\Pi_2^P$-hard.
Our proof will use a reduction from a special case of the quantified Boolean formula problem
({\sc qbf}), as defined next. Let $S \coloneqq \{s_1,\ldots,s_{n_s}\}$ and $T\coloneqq\{t_1,\ldots,t_{n_t}\}$
be two sets of variables, and let $\bar{S}\coloneqq\{\bar{s}_1,\ldots,\bar{s}_{n_s}\}$ and
$\bar{T}\coloneqq\{\bar{t}_1,\ldots,\bar{t}_{n_t}\}$ denote their negated counterparts.
We consider Boolean formulas~$B$ of the form
\[
B \ \ \coloneqq \ \ \forall s_1,\ldots,s_{n_s} \exists t_1,\ldots,t_{n_t}: c_1\land\cdots\land c_{n_c}
\]
where each clause $c_i$ in $C\coloneqq\{c_1,\ldots,c_{n_c}\}$ is a disjunctive combination of at most three literals from $S\cup \bar{S} \cup T \cup \bar{T}$.
Deciding if a formula of this form is true is a $\Pi_2^P$-complete problem~\cite{stockmeyer1976polynomial}.

Consider the undirected graph $G_B\coloneqq (N,A)$ representing $B$, where $N\coloneqq S \cup T \cup  C$ is the set of nodes of $G_B$ and $A \coloneqq \{(c_i,s_j):s_j\in c_i \lor \bar{s}_j \in c_i\} \cup \{(c_i,t_j):t_j\in c_i \lor \bar{t}_j \in c_i\}$
is the set of edges of $G_B$.
Lichtenstein~\cite{lichtenstein1982planar} showed how to transform an instance of {\sc qbf} in polynomial time 
to an equivalent one whose corresponding graph is planar (and of quadratic size). We can use the same technique here. Hence, we may start
our reduction from an Boolean formula $B$ such that $G_B$ is planar.
We call the resulting problem \PlanarQuantCNF.

In the following, we transform an instance of \PlanarQuantCNF
to an instance $\langle V,k,\ell\rangle$ of the Voronoi game problem in the plane
such that $B$ is true if and only if \PlQ has a winning strategy.
\medskip

\begin{figure}
  \centering
  \includegraphics{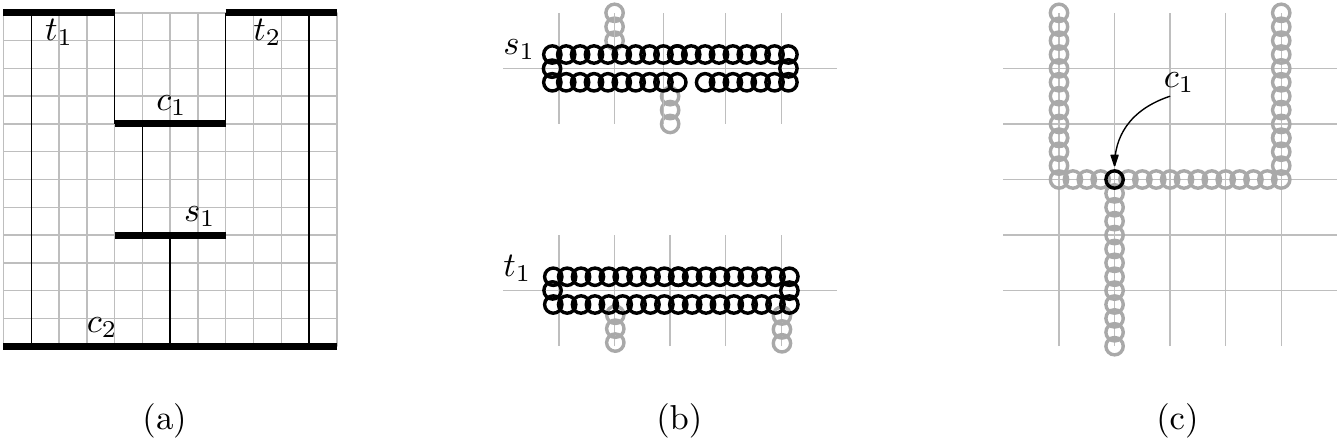}\\
  \caption{{\bf a)} Bar graph of the boolean expression $\forall s_1 \,\exists t_1,t_2: c_1 \land c_2$, where $c_1\coloneqq \bar{s_1} \lor t_1 \lor t_2$ and $c_2 \coloneqq  s_1 \lor t_1 \lor \bar{t_2}$.
  {\bf b)} Representation of variables in the transformed graph. {\bf c)} Representation of a clause in the transformed graph.}
  \label{fig:newgraphongrid}
\end{figure}

The first (and standard) step in our reduction is to construct a specific embedding of the
planar graph~$G_B$. More precisely, we use a \emph{bar graph}, where each node is represented by
a horizontal segment and each edge is represented by a vertical segment; see Figure~\ref{fig:newgraphongrid}a.
Rosenstiehl and Tarjan~\cite{rosenstiehl1986rectilinear} showed that such a representation
always exists.
Before we describe the specific variable, clause, and edge gadgets that we use,
it is useful to make the following observation about when \PlQ wins a certain
voter~$v_i\in V$, when the strategy~$P$ of \PlP is fixed.
Define $D_i$, the \emph{disk of $v_i$} with respect to the given set~$P$, as
the disk with center $v_i$ and radius $\distop(v_i,P)$, that is,
$D_i$ is the largest disk centered at~$v_i$ that has no point
from $P$ in its interior.
The key to our reduction is the following simple observation.
\begin{observation}
Player \PlQ wins a voter $v_i\in V$ against $P$ if and only if she places a point $q$ in the interior of $D_i$.
\end{observation}


Next we describe how to transform the bar graph $G_B$ into a voter set $V$.
Because of the above observation, it will be convenient to imagine that each voter
is surrounded by a disk, and talk about placing disks. Thus, whenever we say we place a disk somewhere, in our construction we actually place a voter at the center of the disk. Later we will then put more voters
that will force most of \PlP's strategy, so that these disks become meaningful.
The nodes of $G_B$ are replaced by the following gadgets:
\begin{description}
\item[Clause gadgets.] Each node $c_i\in C$ is transformed to a single disk as shown in Figure~\ref{fig:newgraphongrid}c.
     Recall that each such node is a horizontal segment, which has three incoming edges from its constituent literals.
    The position of the voter for $c_i$ is (roughly) the point where the middle edge arrives at the
    segment.
\item[Variable gadgets.] Each node $t_i$ of degree $\vdeg(t_i)$ in graph $G_B$ is transformed to a ring of an even number of disks numbered in counterclockwise order starting from any arbitrary disk,
which we call a \emph{closed necklace}, containing at least $2\vdeg(t_i)$ disks as shown in Figure~\ref{fig:newgraphongrid}b.

Each node $s_i$ of degree $\vdeg(s_i)$ in graph $G_B$ is transformed to a ring of disks with one disk missing numbered in counterclockwise order starting from the disk after the missing disk,
which we call an \emph{open necklace}, as shown in Figure~\ref{fig:newgraphongrid}b.
An open necklace has odd size and at least $2\vdeg(s_i)+1$ disks.
We can assume the distance between the two disks in the place where the necklace is open is exactly 1.
Let $D_1$ and $D_2$ be the two disks at distance~1, and let $z_1\in \partial D_1$ and $z_2\in \partial D_2$ be the closest pair of points on the boundaries of these two disks. (Thus $\distop(z_1,z_2)=1$.) Then we place a cluster of $w$ voters at $z_1$ and we place another cluster of $w$ voters at $z_2$,
where $w$ is a suitable number specified later and a \emph{cluster} of voters is simply a number of coinciding voters.

\item[Edge gadgets.]
Each edge $\{c_i,s_j\}$ or $\{c_i,t_j\}$ is replaced by a chain of disks of even length that in one end intersects a pair of consecutive disks in a necklace corresponding to the node $s_j$ or $t_j$, respectively,
such that the first disk has an odd position and the next disk in clockwise order has an even position, and in the other end intersects the disk associated with the clause $c_i$

Similarly, each edge $\{c_i,\bar{s}_j\}$ or $\{c_i,\bar{t}_j\}$ is replaced by a chain of disks  of even length that in one end intersects a pair of consecutive disks in a necklace corresponding to the node $s_j$ or $t_j$, respectively,
such that the first disk has an even position and the next disk in clockwise order has an odd position, and in the other end intersects the disk associated with the clause $c_i$.
\end{description}
In our transformation of $G_B$ to an instance of the Voronoi game, we will use that
$V$ can be a multiset, where we will use multiplicities 1, $w$, $w+1$, $\lfloor w/2\rfloor+1$,
$\lceil w/2\rceil+1$, and $W$ for the voters,
for suitably chosen values $w$ and $W$. We call a cluster of $i$ voters an
\emph{$i$-cluster}.
We denote the multisets containing clusters of 1,$w$, and $W$ voters by $V_1$, $V_w$, and $V_W$, respectively.
The values $w$ and $W$ will be chosen such that they satisfy
\begin{align}
w &= |V_1|+1, \\
W &= |V_w|+|V_1|+1.
\end{align}


\begin{figure}
  \centering
  \includegraphics{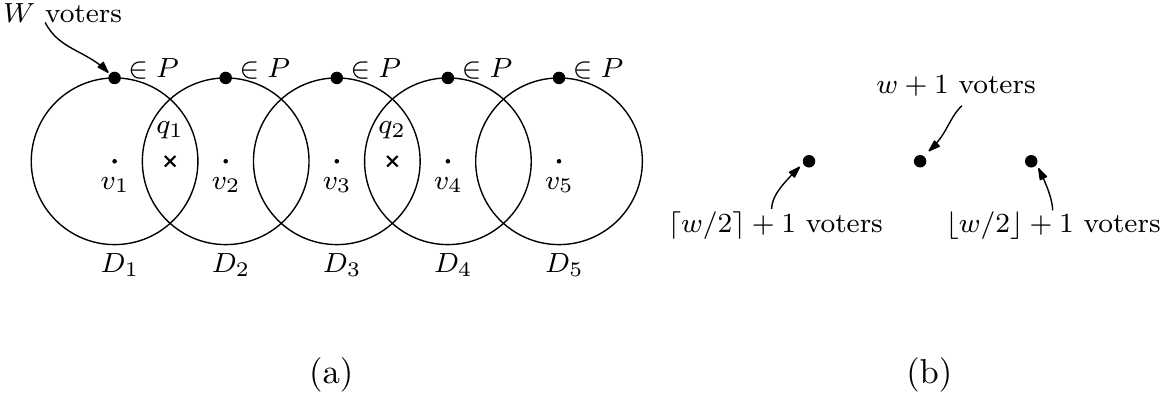}\\
  \caption{{\bf a)} When all the heavy-weight clusters of $W$ voters are chosen by \PlP, the best strategy of \PlQ to win the remaining single voters is to put his points in every other intersection of the disks of the voters.
  {\bf b)} An example of a balancing gadget.}
  \label{fig:newdisksandclusters}
\end{figure}

We have already described the placement of voters in $V_1$ by describing the placement of the disks in defining the gadgets. In our construction, for each voter $v\in V_1$, we consider an imaginary disk $z(v,r)$ which is centered in the position of the voter $v$ and has a radius $r$, which is a real number in range $10\leq r \leq 20$.
Each disk has a $W$-cluster inside it near its boundary which is denoted by $W(v)$.
Two disks $z_1(v_1,r_1)$ and $z_2(v_2,r_2)$ either do not intersect or their interior have a non-empty intersection (see Figure~\ref{fig:newdisksandclusters}a).
Similarly when three disks intersect, their interior have a non-empty intersection. We also described the placement of voters in $V_w$ when we defined the variable gadgets for the variables~$s_i$.

We call this construction \emph{the transformed graph}. Now we add $n'\coloneqq |V_W|-|V_1|+1$ gadgets, which we call \emph{balancing gadgets}, each comprised of three clusters of sizes $\lfloor w/2 \rfloor$+1, $\lceil w/2 \rceil+1$, and $w+1$ far away from the graph of disks and from each other as shown in Figure~\ref{fig:newdisksandclusters}b. Let $V'$ be the multiset of all these voters.

We define $V\coloneqq V_1 \cup V_w \cup V_W \cup V'$, $k\coloneqq \|V_W\|+n'+n_s$, $\ell \coloneqq (|V_1|-n_c-n_s)/2 + n_s + 2n'$, where $\| V_W\|$ denotes the number of distinct points in $V_W$.

\begin{lemma}
\label{lem:hardness}
\PlQ has a winning strategy in $\langle V,k,\ell\rangle$ if and only if $B$ is true.
\end{lemma}
\begin{proof}

First we show that in an optimal strategy, \PlP places his points on the most valuable positions, i.e. $\|V_W\|$ of his points on $W$-clusters, $n'$ of his points on $(w+1)$-clusters, and the remaining $n_s$ points on $w$-clusters where exactly one $w$-cluster is selected from each open necklace. Note that with such a strategy for \PlP, the number of remaining voters for \PlQ is $n_sw+ (w+2)n'+|V_1|$.

For a contradiction, assume \PlP does not select some of the $W$-clusters. In that case \PlP can either move some of his points from balancing gadgets in which he has two or more points or some of his points from the transformed graph which are not exactly on a $W$-cluster to the uncovered $W$-clusters and get a better gain; because each extra point in a balancing gadget has a gain of at most $2w+3<W$ and all the voters of the transformed graph in $V_1$ and $V_w$ together has a gain of at most $|V_w|+|V_1| < W$. Therefore, all the $W$-clusters will be covered by \PlP. Moreover, \PlP must also select all the $(w+1)$ clusters; because by doing this, all the points of \PlQ have a gain of at most $2w+2$ which is less that the gain of an unguarded balancing gadget, i.e. $2w+3$.

It is obvious that \PlP does not place more than one point in each balancing gadget as the gain of each extra point in a balancing gadget is at most $\lceil w/2 \rceil +1$, but the guaranteed gain of a point on a $w$-cluster is $w$. Now, we just need to show that \PlP places all his remaining $n_s$ points on $w$-clusters. This is also easy to see as if more that $n_s$ $w$-clusters are uncovered, \PlQ can easily win $(n_s+1)w+ (w+2)n'$ voters which is more than the total number of remaining voters  $n_sw+ (w+2)n'+|V_1|$ for \PlQ in the case where at most $n_s$ $w$-clusters are uncovered by \PlP.

It is also easy to see that in an optimal strategy, \PlP selects exactly one $w$-cluster in each open necklace. As otherwise, at least one pair of $w$-clusters remains unselected and \PlQ can gain $2w$ voters using just one point, while if exactly one cluster from each pair of $w$-clusters is selected by \PlP, the maximum possible gain of each point of \PlQ is at most $w+1$ voters.

Now, assuming \PlP plays an optimal strategy, the gain of each point of \PlQ is at most
\begin{itemize}
\item $w+1$ voters for a total of $n_s$ points if he puts his point close to an uncovered $w$-cluster,
\item $\lceil w/2 \rceil+1$ or $\lfloor w/2 \rfloor+1$ voters for a total of $2n'$ points if he puts his points in balancing clusters, and
\item $2$ or $3$ voters if he puts his points in the intersection of disks in the transformed graph.
\end{itemize}
This shows that in her optimal strategy \PlQ places all her $2n'$ points in the balancing clusters to win all the remaining voters there, and the remaining $(|V_1|-n_c-n_s)/2 + n_s$ points of \PlQ will be placed on the transformed graph.

\begin{figure}
  \centering
  \includegraphics{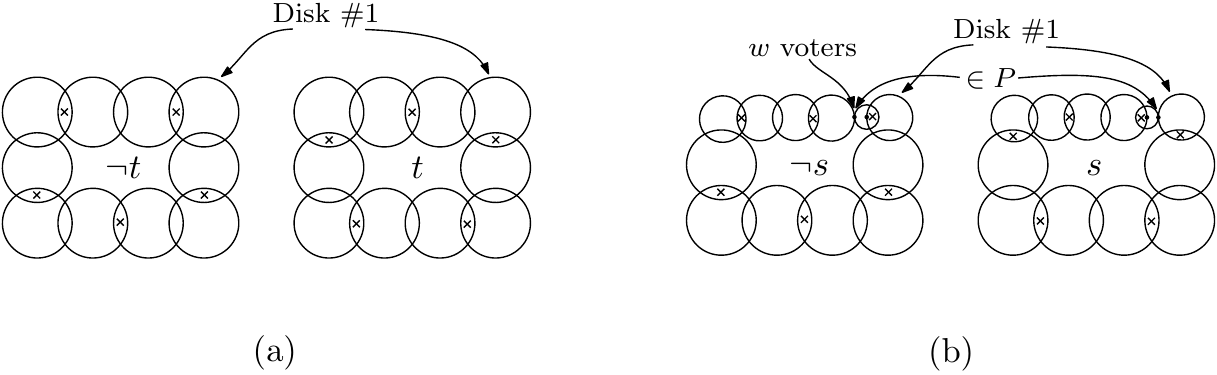}\\
  \caption{{\bf a)} The two different ways of winning all voters for \PlQ. {\bf b)} \PlP can force \PlQ's game.}
  \label{fig:newtruefalse}
\end{figure}

To win the remaining voters of a necklace efficiently, \PlQ has two choices. She can either put her points in odd intersections or in even intersections of the necklace as shown in Figure~\ref{fig:newtruefalse}a. In both cases the number points she spends is equal to half the number of disks in the necklace. To win the remaining voters of an open necklace, \PlQ has just one choice. We note that the $w$-cluster selected by \PlP effectively changes the open necklace to a closed necklace with one more disk (see Figure~\ref{fig:newtruefalse}b). Depending on the $w$-cluster chosen by \PlP, player \PlQ must put her points in odd or even intersections as shown in Figure~\ref{fig:newtruefalse}b. In either case \PlQ spends $(m+1)/2$  points, where $m$ is the number of disks in the open necklace. To win all the remaining voters in a chain of $m$ disks, \PlQ should place some points in every other intersection of the chain in odd positions and he spends $m/2$ points. Alternatively, she can place her $m/2$ points in even positions and also win the point in the clause on the one end of the chain only if the other end of the chain has already been covered by one of his points in the (open) necklace.

The placement of points by \PlQ in odd intersections of a necklace corresponds to a true value for the associated $t_j$, and placement in even intersections corresponds to a false value for $t_j$. Similarly, selection of the starting $w$-cluster of an open necklace by \PlP corresponds to a true value for the associated $s_j$, and selection of the ending $w$-cluster corresponds to a false value for $s_j$. When \PlQ has a choice to put his points in even intersections of a chain and win the voter of the associated clause at the end of the chain, it corresponds to a true value for the associated literal of the (open) necklace at the other end of the chain which gives the clause a true value. Intuitively, it is clear that given an optimal strategy $P$, \PlQ should use all his points to win the remaining voters and he can win the voter associated with each clause if and only if she can satisfy that clause.

More formally, given an optimal strategy $P$ of \PlP, we set $s_j$ to false if the starting $w$-cluster in clockwise order in the open necklace of the gadget for $s_j$ is selected in $P$ and set it to true otherwise; we call this set of values $\fs(P)$. By using the discussion in the two previous paragraphs, we can easily verify that if there exists an assignment $\ft$ for the variables $t_j (j=1,\dots,n_t)$ such that $\langle \fs(P), \ft\rangle$ satisfies $c_1\land\cdots\land c_{n_c}$, then \PlQ can win all the $n_sw+ (w+2)n'+|V_1|$ remaining voters.
Conversely, it is not hard to see that if \PlQ can win all the $n_sw+ (w+2)n'+|V_1|$ remaining voters, then there exists an assignment $\ft$ for $t_j$s such that $B$ is true.

Considering the fact that the number of remaining voters is exactly one more than the number of voters directly won by $P$, i.e.
$$\left(n_sw+ (w+2)n'+|V_1|\right)-\left(\|V_W\|+n_sw+ (w+1)n'\right) = 1,$$
\PlQ has a winning strategy in $\langle V,k,\ell\rangle$ if and only if $B$ is true.
\end{proof}

\label{sec:poly-reduction}
Lemma~\ref{lem:hardness} is not enough to show that the Voronoi game is $\Sigma_2^P$-hard. We need to show that the reduction can be done in polynomial time and therefore the resulting Voronoi game has polynomial size.


We can easily generate balancing gadgets. Therefore we focus on drawing the transformed graph. As stated earlier, in order to define the exact position of clusters in the transformed graph we use a method devised by Rosenstiehl and Tarjan in \cite{rosenstiehl1986rectilinear} to draw a planar graph on a grid of size $O(n)\times O(n)$ where each node is represented by a horizontal line segment and each edge is represented by a vertical line segment (see Figure~\ref{fig:newgraphongrid}a). As the graph $G_B$ is planar, we can draw it using this method. An (open) necklace associated with a variable $t_j$ or $s_j$ can be easily drawn as shown in Figure~\ref{fig:newgraphongrid}b. With a big enough cell-size for the grid, say 1000, we can adjust the disk sizes to get the desired properties for the intersections and also for the distance 1 between two ending disks in an open necklace. As a node associated with a clause has a degree of at most three, it can be drawn as in Figure~\ref{fig:newgraphongrid}c, and similarly with a big enough cell-size for the grid, we can adjust the size of the disks to get an even number of disks for each chain. The $W$-clusters can be placed at a distance of at most $0.1$ at a rational point near the border of each disk, and placing $w$-clusters is also trivial.

As the total number of nodes in $G_B$ is $O(n_s+n_t+n_c)$, we need a grid of size $O(n_s+n_t+n_c) \,\cdot\, O(n_s+n_t+n_c)$, and in worst case we have a constant number of disks on each edge of this grid. Therefore, the number of disk of $|V_1|$ is bounded from above by the size of the grid, i.e. $|V_1|\in O((n_s+n_t+n_c)^2)$, which means $|V|\in O(|V_1|W+n_sw+|V_1|+(2w+3)n') = O((n_s+n_t+n_c)^5)$. The following theorem is the result of this discussion.
\begin{theorem}
The Voronoi game problem is $\Sigma_2^P$-hard for $d>1$.
\end{theorem}

We note that the final construction is rectilinear and therefore this reduction also works for all the $L_p$ norms as disks in the final construction can be replaced by $L_p$-disks and all the required properties still hold. Moreover, as this reduction is done for the two-dimensional case, this hardness result is also true for the $L_\infty$ norm, because in $\Reals^2$, the disks of the $L_\infty$ norm are similar to those in the $L_1$ norm, just rotated by $\pi/4$.

\section{Containment in \texorpdfstring{$\exists\forall \Reals$}{Sigma-2-Real} and the algorithm for \texorpdfstring{$d\geq 2$}{d>1}}
We now consider the one-round discrete Voronoi game in the $L_p$-norm,
for some arbitrary~$p$. Then a strategy $P=\{p_1,\ldots,p_k\}$ can win a voter $v\in V$
against a strategy $Q=\{q_1,\ldots,q_\ell\}$ if and only if the following
Boolean expression is satisfied:
$$
\win(v)\coloneqq\bigvee_{i\in[k]} \bigwedge_{j\in[\ell]} (\distop_p(p_i,v))^p \leq (\distop_p(q_j,v))^p, \\
$$
where $\distop_p$ is the $L_p$-distance.
This expression has $k\ell$ polynomial inequalities of degree $p$. The strategy~$P$ is winning
if and only if the majority of the expressions $\win(v_1),\ldots, \win(v_n)$ are true.
Having a majority function $\maj$ that evaluates to true if at least half of its parameters evaluates to true, player~\PlP has a winning strategy if and only if
\begin{multline*}
\exists x_1(p_1), \ldots, x_d(p_1),\ldots,x_1(p_k), \ldots, x_d(p_k)\\ \forall x_1(q_1), \ldots, x_d(q_1),\ldots,x_1(q_\ell), \ldots, x_d(q_\ell): \quad \maj(\win(v_1),\ldots,\win(v_n)) \\
\end{multline*}
is true, where $x_i(\cdot)$ denotes the $i$-th coordinate of a point.

Ajtai~\etal \cite{ajtai1983sorting} show that it is possible to construct a sorting network, often called the AKS sorting network, composed of comparison units configured in $c\cdot \log n$ levels, where $c$ is a constant and each level contains exactly $\lfloor n/2\rfloor$ comparison units. Each comparison unit takes two numbers as input and outputs its input numbers in sorted order. Each output of a comparison unit (except those on the last level) feeds into exactly one input of a comparison unit in the next level, and the input numbers are fed to the inputs of the comparison units in the first level. The outputs of the comparison units in the last level (i.e., the outputs of the network) give the numbers in sorted order.

Using AKS sorting networks we can construct a Boolean formula of size $O(n^c)$ for some constant $c$ that tests if the majority of its $n$ inputs are true as follows. Assuming the boolean value \emph{false} is smaller than the boolean value \emph{true} value, we make an AKS sorting network that sorts $n$ boolean values. This is possible using comparison units that get $p$ and $q$ as input, and output $p \land q$ and $p \lor q$. It is not hard to verify that the $\lceil n/2 \rceil$-th output of the network is equal to the value of the majority function on the input boolean values. By construction, we can write the Boolean formula representing the value of this output as \emph{logical and} ($\land$) and \emph{logical or} ($\lor$) combination of the input boolean values, and the size of the resulting formula is $O(n^c)$.

Thus we can write $\maj(\win(v_1),\ldots,\win(v_n))$
as a Boolean combination of $O(n^ck\ell)$ polynomial inequalities of degree~$p$,
where each quantified block has $kd$ and $\ell d$ variables respectively.
Basu~\etal~\cite{basu1996combinatorial} gave an efficient algorithm
for deciding the truth of quantified formulas. For our formula
this gives an algorithm with $O((n^ck\ell)^{(kd+1)(\ell d +1)}p^{k\ell d^2})$
running time to decide if \PlP has a winning strategy for a
given instance~$\langle V,k,\ell\rangle$ of the Voronoi game problem.
Note that this is polynomial when $k$, $\ell$ and $d$ are constants.

%
For the $L_{\infty}$ norm, we can define $F(v)$ as follows:
$$
F(v)\coloneqq\bigvee_{i\in[k]} \bigwedge_{j\in[\ell]} \bigvee_{s'\in[d]} \bigwedge_{s\in[d]} |x_s(p_i)-x_s(v)| \leq |x_{s'}(q_j) - x_{s'}(v)|,
$$
By comparing the squared values instead of the absolute values, we have a formula which demonstrates that even with the $L_{\infty}$ norm, the problem is contained in $\exists\forall \Reals$ and there exists an algorithm of complexity $O((n^ck\ell d^2)^{(kd+1)(\ell d +1)}2^{k\ell d^2})$ to solve it.

\begin{theorem}
The one-round discrete Voronoi game $\langle V, k, \ell \rangle$ in $\Reals^d$ with the $L_p$ norm is contained in $\exists\forall \Reals$. Moreover, for fixed $k$, $\ell$, $d$ there exists an algorithm that solves it in polynomial time.
\end{theorem}

De Berg~\etal~\cite{berg2018faster} introduced the notion of personalized preferences.
More precisely, given a natural number $p$, assuming each axis defines an aspect of the subject voters are voting for, the voter $v_i$ gives different weights to different axes, and $v_i$ has a weighted $L_p$ distance $(\sum_{j\in[d]}w_{ij}(x_j(p)-x_j(v_i))^p)^{1/p}$ from any point $p\in \Reals^d$. For the weighted $L_\infty$ distance, $v_i$ is at distance $\max_{j\in[d]}(w_{ij}|x_j(p)-x_j(v_i)|)$ from any point $p\in\Reals^d$. This approach also works when voters have personalized preferences.

\section{Concluding remarks}
We presented the first polynomial-time algorithm for the one-round discrete
Voronoi game in $\Reals^1$. The algorithm is quite intricate, and it
would be interesting to see if a simpler (and perhaps also faster) algorithm is possible.
Finding a lower bound for the 1-dimensional case is another open problem.

We also showed that the problem is $\Sigma_2^P$-hard in $\Reals^2$.
Fekete and Meijer~\cite{fekete2005one} conjectured that finding an optimal strategy
for the multi-round continuous version of the Voronoi game is PSPACE-complete.
We conjecture that in the multi-round version of the discrete version,
finding an optimal strategy  is PSPACE-hard as well.
Note that using the algebraic method presented in this paper,
it is easy to show that this problem is contained in PSPACE.
While the algebraic method we used is considered a standard technique,
it is, as far as we know, the first time this method is combined with
polynomial-size boolean formulas for the majority function.
We think it should be possible to apply this combination
to other problems as well.



\bibliography{thesis}

\end{document}